\documentclass[a4paper]{article}

\usepackage[english]{babel}
\usepackage[utf8x]{inputenc}
\usepackage[T1]{fontenc}
\usepackage{graphicx}
\usepackage{caption, subcaption}
\usepackage[title]{appendix}
\usepackage{tikz}
\usepackage{pgfplots}
\usepackage{cutwin}
\graphicspath{ {images/} }
\usetikzlibrary{intersections,decorations.pathreplacing,decorations.markings,calc,angles,quotes,arrows.meta,pgfplots.fillbetween,patterns}

\usepackage[a4paper,top=2.8cm,bottom=3cm,left=2.7cm,right=2.7cm,marginparwidth=1.75cm]{geometry}

\usepackage{amsmath,yhmath,amsfonts}
\usepackage{amsthm}
\usepackage{amssymb}
\usepackage[hidelinks]{hyperref}
\usepackage{cleveref}
\usepackage{epstopdf}
\usepackage{comment}
\usepackage{todonotes}
\usepackage{float}
\usepackage[sort&compress,numbers]{natbib}
\newtheorem{thm}{Theorem}% 	 	[section]
\newtheorem{lem}[thm]{Lemma}%[section]
\theoremstyle{definition}
\newtheorem{defn}[thm]{Definition}
\newtheorem{rmk}[thm]{Remark}%[section]
\numberwithin{equation}{section}
\numberwithin{thm}{section}

\newcommand{\C}{\mathcal{C}}

\newcommand{\D}{\mathcal{D}}

\newcommand{\Hc}{\mathcal{H}}

\newcommand{\R}{\mathbb{R}}
\newcommand{\Rc}{\mathcal{R}}
\renewcommand{\S}{\mathcal{S}}

\newcommand{\V}{\mathcal{V}}
\newcommand{\Z}{\mathbb{Z}}
\newcommand{\p}{\partial}
\renewcommand{\epsilon}{\varepsilon}
\newcommand{\dx}{\: \mathrm{d}}

\newcommand{\uf}{\mathfrak{u}}
\newcommand{\Bf}{\mathfrak{B}}
\newcommand{\Cf}{\mathfrak{C}}

\renewcommand{\v}{\mathbf{v}}
\newcommand{\frm}{\mathrm{f}}
\newcommand{\trm}{\mathrm{t}}

\newcommand{\nm}{\noalign{\smallskip}}
\newcommand{\ds}{\displaystyle}
\newcommand{\iu}{\mathrm{i}\mkern1mu}

\DeclareMathOperator*{\argmax}{argmax}
\newcommand{\ddp}[2]{\frac{\partial#1}{\partial#2}}

\makeatletter
\newcommand{\neutralize}[1]{\expandafter\let\csname c@#1\endcsname\count@}
\makeatother

\title{Convergence rates for defect modes in large finite resonator arrays}
\author{
	Habib Ammari\thanks{\footnotesize Department of Mathematics,
		ETH Zurich, Zurich, Switzerland (habib.ammari@math.ethz.ch).}\and Bryn Davies\thanks{\footnotesize Department of Mathematics, Imperial College London, London, UK (bryn.davies@imperial.ac.uk).} \and Erik Orvehed Hiltunen\thanks{\footnotesize Department of Mathematics, Yale University, New Haven, USA (erik.hiltunen@yale.edu).}}
\date{}

\begin{document}

	\maketitle
	
	\begin{abstract}
	We show that defect modes in infinite systems of resonators have corresponding modes in finite systems which converge as the size of the system increases. We study the generalized capacitance matrix as a model for three-dimensional coupled resonators with long-range interactions and consider defect modes that are induced by compact perturbations. If such a mode exists, then there are elements of the discrete spectrum of the corresponding truncated finite system that converge to each element of the pure point spectrum. The rate of convergence depends on the dimension of the lattice. When the dimension of the lattice is equal to that of the physical space, the convergence is exponential. Conversely, when the dimension of the lattice is less than that of the physical space, the convergence is only algebraic, because of long-range interactions arising due to coupling with the far field.
	\end{abstract}
	\vspace{0.5cm}
	\noindent{\textbf{Mathematics Subject Classification (MSC2010):} 35J05, 35C20, 35P20.
		
	\vspace{0.2cm}
	
	\vspace{0.5cm}

\section{Introduction}
Much of the physical literature concerning wave propagation in periodic media relies on a believable but highly non-trivial piece of logic. That is, researchers want to be able to relate the spectral properties of infinite periodic structures with truncated, finite versions of the same material. The motivation for this is that infinite periodic structures can be described very concisely using Floquet-Bloch analysis \cite{kuchment1993floquet}. However, finite, truncated versions of the structure are often required when it comes to either numerical or physical experiments. It is perfectly plausible that the two structures should behave similarly, particularly away from the edges of the truncated structure and especially when the truncated structure is very large. However, a precise convergence theory relating the spectra of these two quite different differential operators is, in general, yet to be developed.

The many interesting phenomena that occur at the edges of periodic arrays have been studied in some detail \cite{hills1965semiinfinite}. For example, there is a tendency for wave energy to be localized to the edges of the structure, taking the form of surface waves \cite{vinogradova2019full, linton2007scattering}. This is an example of an \emph{edge effect} and highlights that there will always be fundamental differences between how infinite and truncated structures interact with waves. Another important question that has been explored in this field, and is intimately related to the results presented in this work, is the extent to which waves incident on the edge of a truncated periodic structure can excite Bloch waves in the structure (thus, replicating the behaviour of its infinite counterpart) \cite{tymis2014scattering, joseph2015reflection, thompson2018directI}.

The central question of this work is the extent to which the resonant spectra of infinite and truncated structures can be related. We will focus on localized modes which decay quickly outside of some compact region, meaning they are less severely affected by edge effects. Additionally, localized modes are the eigenmodes of interest for many wave guiding applications. Existing results have shown that in certain one- and two-dimensional systems, any \emph{defect mode} of the infinite structure will have a corresponding mode in the truncated structure converging to the defect mode as the size tends to infinity \cite{lin2013resonances, lin2016perturbation, lu2022defect, lin2015twodim}. A defect mode is a mode that is created by making a perturbation to introduce a defect to the periodic structure. Such a mode is characterized by being spatially localized (in the sense that it decays quickly enough to be square integrable along the axis or axes of periodicity) and having an eigenfrequency that belongs to the pure point spectrum of the perturbed periodic operator. The terminology ``pure point spectrum'' and ``defect mode eigenfrequency'' are preferred by spectral analysts and wave physicists, respectively, and we will use them somewhat interchangeably here. The rest of the spectrum will typically be composed of the \emph{continuous spectrum}, which corresponds to the Bloch modes that propagate through the material without decaying. 

In previous works, it was shown that the convergence of defect mode eigenfrequency to the pure point spectrum of the infinite periodic operator was exponential with respect to the size of the truncated array \cite{lin2013resonances, lin2016perturbation, lu2022defect, lin2015twodim}. These results concerned either one-dimensional systems or two-dimensional systems with two-dimensional lattices. In this work, we study the \emph{generalized capacitance matrix}, which is a dense resonator model that includes long-range interactions \cite{ammari2021functional}. This model gives a leading-order characterisation of the resonant modes of a three-dimensional scattering problem with high-contrast resonators. However, it can be viewed more generally as a canonical model for coupled resonators. In this setting, we will prove that any defect mode eigenfrequency of the infinite structure has a sequence of eigenvalues of the truncated structures converging to it. This expands on the one- and two-dimensional models explored previously, by showing that this convergence is algebraic when the dimension of the periodic lattice is less than the dimension of the differential operator (which is three, in this work). This is due to the long-range interactions that arise since waves are able to radiate in the ``spare'' dimensions and couple with the far field. When the lattice is three-dimensional, with no spare dimensions, we see the same exponential convergence as for one-dimensional lattices in one-dimensional differential problems \cite{lin2013resonances, lin2016perturbation, lu2022defect} and two-dimensional lattices in two-dimensional differential problems \cite{lin2015twodim}. We believe that similar behaviour will be observed in other multi-dimensional differential systems and dense matrix models. We emphasise that the results of this work hold in the Hermitian case of real, positive, material parameters.

This paper is split into three main parts. In \Cref{sec:prelims}, we introduce the matrix model (the generalized capacitance matrix) that we will study and prove some elementary properties that lay the foundations for the subsequent analysis. \Cref{sec:purep} contains the main results of this work, which show that the truncated structures have eigenfrequencies that converge to the pure point spectrum of the infinite structure (and characterise the rate of this convergence). Finally, in \Cref{sec:cont}, we briefly present numerical evidence for the convergence of the truncated spectra to the continuous spectrum. This will be handled in more detail in \cite{ammari2023spectral}, where we take advantage of the constructive nature of the generalized capacitance matrix to prove the convergence of the Bloch modes.

\section{The generalized capacitance matrix model} \label{sec:prelims}

In this section, we will introduce the generalized capacitance matrix model that will be the object of this study. Its definition uses layer potentials to capture the (potentially complex) shapes of the resonators. In \Cref{sec:asymptotics} give a detailed presentation of asymptotic results showing how this model can be deduced from a subwavelength resonance problem with a system of high-contrast resonators. Finally, we will prove a convergence result for the capacitance coefficients that will be the basis of the theorems in subsequent sections.

\subsection{Definition}

We study a system of periodically repeated resonators in a lattice in $\R^3$.  We take lattice vectors $l_1, \dots, l_{d}\in \R^3$, where $0<d\leq3$, and let $\Lambda$ denote the lattice generated by these vectors. In other words,
$$\Lambda := \left\{ m_1 l_1+\dots+m_{d} l_{d} ~|~ m_i \in \Z \right\}. $$
At this point, we remark that there are three possible cases: $d=1$, corresponding to a \emph{chain} of resonators; $d=2$, corresponding to a \emph{screen} of resonators; or $d=3$, corresponding to a \emph{crystal} of resonators. In the case $d<3$, the resonator structure is bounded in the direction(s) perpendicular to the lattice, and the waves are radiating outwards. For simplicity, we assume that the lattice is aligned with the first $d$ coordinate axes.

We take $Y\subset \R^3$ to be a single unit cell,
$$Y = \begin{cases}
	\{c_1 l_1 + x_2 e_2 + x_3e_3 \mid 0\leq c_1\leq 1, x_{2},x_3 \in \R \}, & d=1, 
	\\
	\{c_1 l_1 + c_2 l_2 + x_3e_3 \mid 0\leq c_1,c_2\leq 1, x_3 \in \R \}, & d=2,
	\\
	\{c_1 l_1 + c_2 l_2 + c_3l_3 \mid 0\leq c_1,c_2,c_3\leq 1 \}, &d=3.
\end{cases}$$
The resonators studied in this work are bounded inclusions of a contrasting material inside some background, as depicted in \Cref{fig:notation}. We let $D\subset Y$ be a collection of $N$ resonators contained in $Y$
$$D = \bigcup_{i=1}^N D_i,$$
where $D_n$ are disjoint, bounded domains in $Y$ with boundary $\p D_i \in C^{1,s}$ for $s>0$. In the periodic lattice, we let $D_i^m = D_i + m,$ for $m\in \Lambda$, and then denote the full lattice as 
$$\D=\bigcup_{m\in\Lambda} \bigcup_{i=1}^N D_i^m.$$

We will define a finite system of resonators resulting from truncation of the periodic lattice. Let $I_r \subset \Lambda $ be all lattice points within distance $r$ from the origin
$$I_r = \{m \in \Lambda \mid |m| < r \}.$$
We define the finite collection of resonators $\D_\frm = \D_\frm(r)$ as 
\begin{equation} \label{eq:finitelattice}
\D_\frm(r) = \bigcup_{m\in I_r} D + m.
\end{equation}
In this setting, $\D_\frm$ is a finite lattice where $D$ is the single, repeated unit. The goal is to clarify in which sense the spectral properties of a finite, but large, lattice can be approximated by the corresponding infinite one.

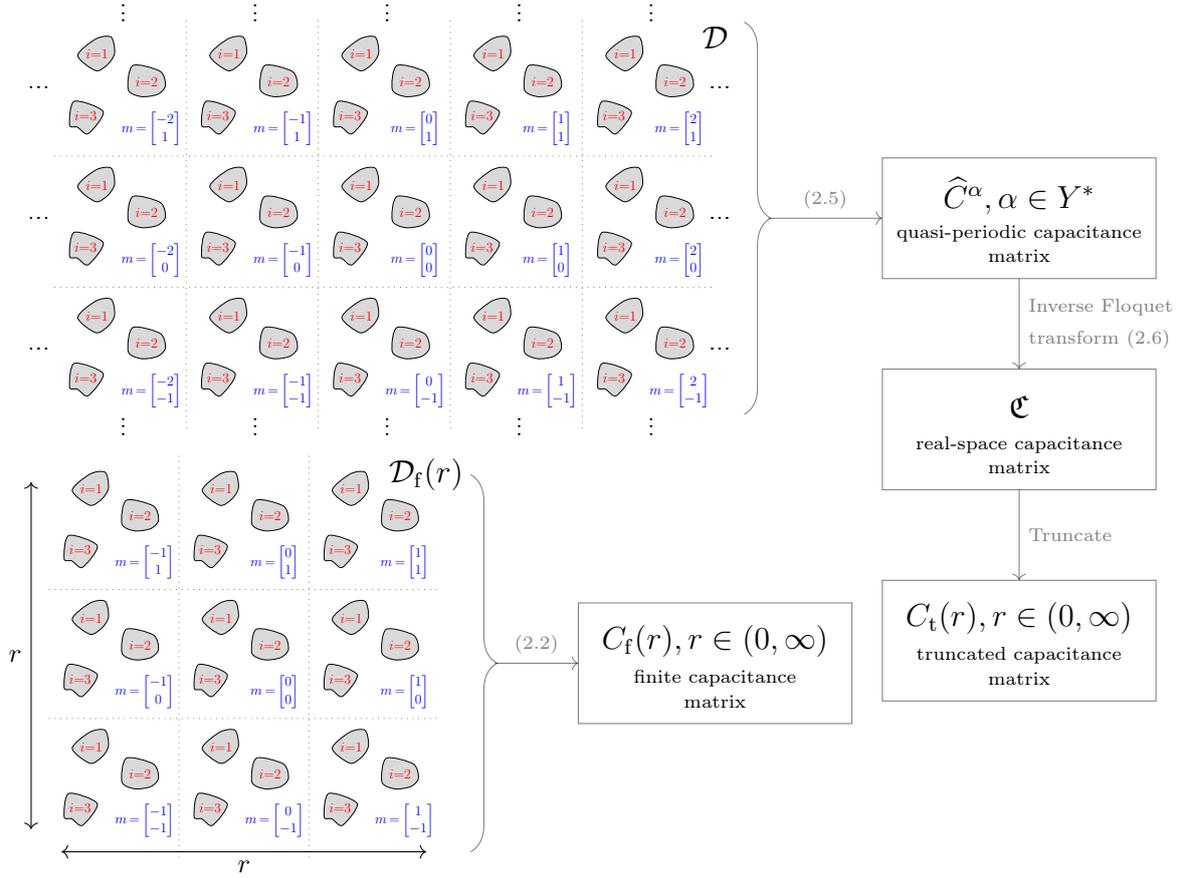
\begin{figure}
\centering
\begin{tikzpicture}[scale=2]
		
	\node[inner sep=0pt] at (-0.2,0)  {\includegraphics[height=5.9cm]{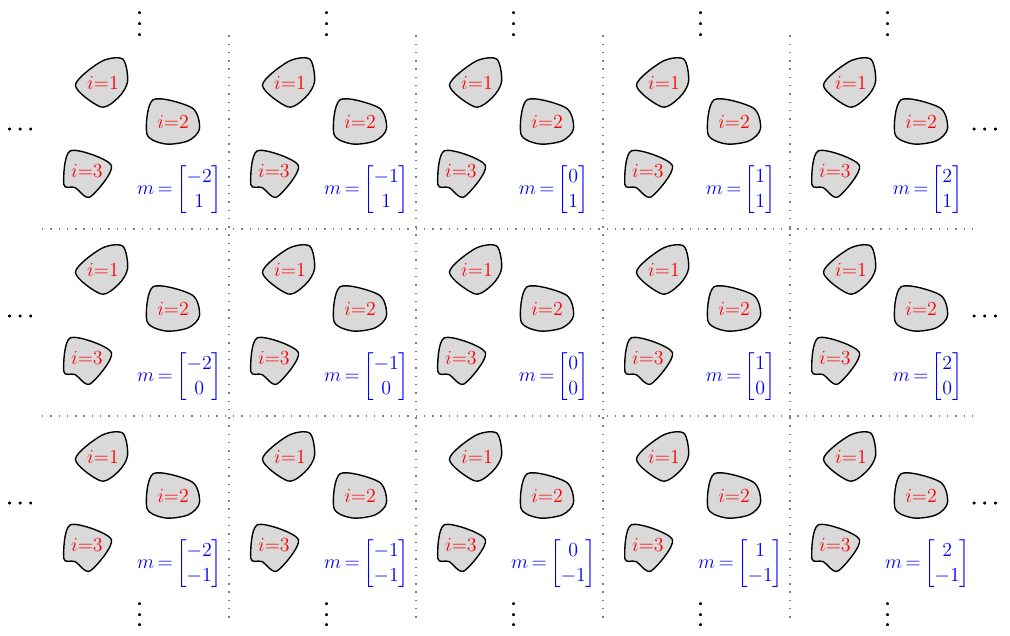}};
		
	\begin{scope}[xshift=-0.9cm,yshift=-2.9cm]
	\node[inner sep=0pt] at (-0.2,0)  {\includegraphics[height=5.4cm]{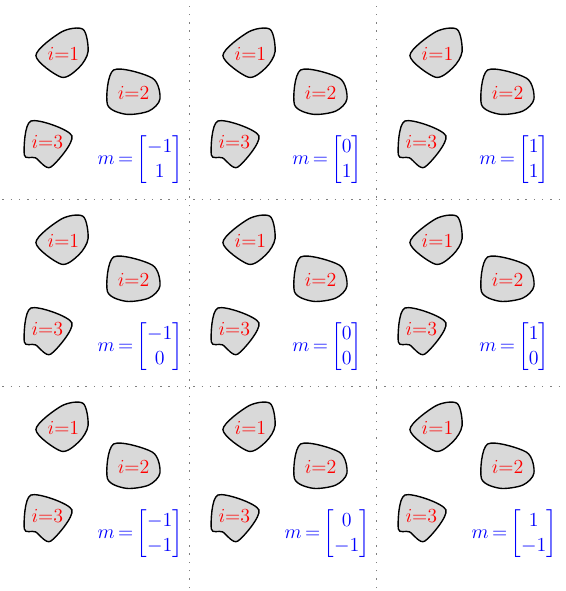}};
	\draw[<->] (-1.4,-1.3) -- (1,-1.3) node[midway, below]{$r$};
	\draw[<->] (-1.6,-1.15) -- (-1.6,1.15) node[midway, left]{$r$};
	\draw [decorate,opacity=0.5,rotate=180,decoration={brace,amplitude=10pt},xshift=-4pt,yshift=0pt](-1.15,-1.2) -- (-1.15,1.3);
	\draw[->,opacity=0.5] (1.45,-0.05) -- (2,-0.05) node[midway,above]{\scriptsize \eqref{eq:Cfinite}};
	\node at (1,1.2) {\large $\D_\frm(r)$};
	\end{scope}

	\draw [decorate,opacity=0.5,rotate=180,decoration={brace,amplitude=10pt},xshift=-4pt,yshift=0pt](-2.05,-1.3) -- (-2.05,1.3);
	\draw[->,opacity=0.5] (2.35,0) -- (3.1,0) node[midway,above]{\scriptsize \eqref{eq:Calpha}};
	\node at (2,1.2) {\large $\D$};
	
	\node at (4,0.15) {\large $\widehat{C}^\alpha, \alpha \in Y^*$};
	\node at (4,-0.1) {\scriptsize quasi-periodic capacitance};
	\node at (4,-0.25) {\scriptsize matrix};
	\draw[opacity=0.5] (3.1,-0.4) rectangle (4.9,0.4);
	
	\draw[->,opacity=0.5] (4,-0.4) -- (4,-1) node[midway,right, align=left]{\scriptsize Inverse Floquet \\ \scriptsize transform \eqref{eq:Crealspace}};
	
	\node at (4,-1.25) {\large $\Cf$};
	\node at (4,-1.5) {\scriptsize real-space capacitance};
	\node at (4,-1.65) {\scriptsize matrix};
	\draw[opacity=0.5] (3.1,-1.8) rectangle (4.9,-1);
	
	\draw[->,opacity=0.5] (4,-1.8) -- (4,-2.4) node[midway,right, align=left]{\scriptsize Truncate};
	\draw[opacity=0.5] (3.1,-3.2) rectangle (4.9,-2.4);
	
	\node at (4,-2.65) {\large $C_\trm(r), r\in(0,\infty)$};
	\node at (4,-2.9) {\scriptsize truncated capacitance};
	\node at (4,-3.05) {\scriptsize matrix};
	
	\begin{scope}[yshift=-0.15cm]
	\draw[opacity=0.5] (1.1,-3.2) rectangle (2.9,-2.4);
	\node at (2,-2.65) {\large $C_\frm(r), r\in(0,\infty)$};
	\node at (2,-2.9) {\scriptsize finite capacitance};
	\node at (2,-3.05) {\scriptsize matrix};
	\end{scope}
	\end{tikzpicture}
	\caption{This work studies the convergence of the eigenfrequencies of defect modes in a truncated periodic material to the spectrum of the corresponding infinite material. We use capacitance matrices as a canonical model for many-body scattering of time-harmonic waves. The aim of this work is to show how eigenvalues of the finite capacitance matrix $C_\frm(r)$ converge to those of the real-space capacitance matrix $\Cf$. The calligraphic font for $\Cf$ denotes the fact that this is an infinite matrix. Our strategy is to compare the spectrum of $C_\frm(r)$ with the truncated capacitance matrix $C_\trm(r)$, which is obtained by truncating all but a finite $O(r)$ number of rows in $\Cf$, before letting $r\to\infty$. Throughout this work, we use the block matrix notation $(C^{\color{blue} mn})_{\color{red} ij}$ to refer to the $\color{red} i,j\in\{1,\dots,N\}$ entry of the $\color{blue} n,m\in\Lambda$ block in a matrix $C$.} \label{fig:notation}
\end{figure}

Next, we define the capacitance matrix model for the resonator systems described above. This can be viewed as a discrete canonical model for resonator systems with long-range coupling. For reference, in \Cref{sec:asymptotics} we present details on the differential (Helmholtz) problem on $\D$ which is well-approximated by the capacitance model in the subwavelength regime (\Cref{thm:asymp}). We let $G$ be the Green's function for Laplace's equation in three dimensions:
$$G(x) = -\frac{1}{4\pi|x|}.$$
Given a bounded domain $\Omega \subset \R^3$, we then define the \emph{single layer potential}  $\mathcal{S}_\Omega: L^2(\p \Omega) \to H^1(\p \Omega)$ as
$$\mathcal{S}_\Omega[\varphi](x) := \int_{\partial \Omega} G(x-y) \varphi(y) \dx\sigma(y),\quad x\in \p \Omega.$$
Here, $L^2(\partial\Omega)$ is the space of all square-integrable functions on the boundary $\partial\Omega$ and $H^1(\partial\Omega)$ is the subset of elements of $L^2(\partial\Omega)$ which have weak first derivatives that are also square integrable. An important property of $\S_\Omega$ is that it is known to be invertible \cite{ammari2009layer}. 
\begin{defn}[Capacitance coefficients for a finite lattice]
For a finite lattice, as defined in \eqref{eq:finitelattice}, we define the capacitance coefficients as 
\begin{equation} \label{eq:Cfinite}
(C^{mn}_\frm)_{ij}(r) = \int_{\p D_i^m} \S_{\D_\frm}^{-1}[\chi_{\p D_j^n}]  \dx \sigma,
\end{equation}
for $1\leq i,j\leq N$ and $m,n \in I_r$. Here, $\chi_A$ is used to denote the characteristic function of a set $A$. We have explicitly indicated the dependence on the size $r$ of the truncated lattice. For $m,n \in I_r$, we observe that $C^{mn}_\frm(r)$ is a matrix of size $N\times N$, while the block matrix $C_\frm = (C^{mn}_\frm)$ is a matrix of size  $N|I_r|\times N|I_r|$.
\end{defn}

We next define the capacitance coefficients for the infinite lattice. We begin by defining the dual lattice $\Lambda^*$ of $\Lambda$ as the lattice generated by $\alpha_1,...,\alpha_{d}$ satisfying $ \alpha_i\cdot l_j = 2\pi \delta_{ij}$ for $i,j = 1,...,d,$  and whose projection onto the orthogonal complement of $\Lambda$ vanishes. We define the \emph{Brillouin zone} $Y^*$ as $Y^*:= \big(\R^{d}\times\{\mathbf{0}\}\big) / \Lambda^*$, where $\mathbf{0}$ is the zero-vector in $\R^{3-d}$. We remark that $Y^*$ can be written as $Y^*=Y^*_d\times\{\mathbf{0}\}$, where  $Y^*_d$ has the topology of a torus in $d$ dimensions.

When $\alpha\in  Y\setminus \{0\}$, we can define the quasi-periodic Green's function $G^{\alpha}(x)$ as
\begin{equation}\label{eq:xrep}
	G^{\alpha}(x) := \sum_{m \in \Lambda} G(x-m)e^{\iu \alpha \cdot m}.
\end{equation}
The series in \eqref{eq:xrep} converges uniformly for $x$ in compact sets of $\R^d$, with $x\neq 0$ and $\alpha \neq 0$. Given a bounded domain $\Omega \subset Y$, we can then define the \emph{quasi-periodic} single layer potential  $\mathcal{S}_\Omega^{\alpha}: L^2(\p \Omega) \to H^1(\p \Omega)$ as
\begin{equation}
\mathcal{S}_\Omega^{\alpha}[\varphi](x) := \int_{\partial \Omega} G^{\alpha} (x-y) \varphi(y) \dx\sigma(y),\quad x\in \p \Omega.
\end{equation}
It is well-known that  $\mathcal{S}_\Omega^{\alpha}: L^2(\p \Omega) \to H^1(\p \Omega)$ is an invertible operator \cite{ammari2009layer}.
\begin{defn}[Capacitance coefficients for an infinite lattice]
For $\alpha \in Y^*$ and for $1\leq i,j\leq N$, the quasi-periodic capacitance matrix (``dual-space'' representation) is the $N\times N$-matrix defined as 
\begin{equation} \label{eq:Calpha}
\widehat{C}_{ij}^\alpha = \int_{\p D_i} (\S_D^\alpha)^{-1}[\chi_{\p D_j}] \dx \sigma.
\end{equation}
For $1\leq i,j\leq N$, we can then define the ``real-space'' capacitance coefficients at the lattice point $m$ by
\begin{equation} \label{eq:Crealspace}
C_{ij}^m = \frac{1}{|Y^*|}\int_{Y^*} \widehat{C}_{ij}^\alpha e^{-\iu \alpha\cdot m}\dx \alpha.
\end{equation}
Here, $C_{ij}^0$ corresponds to the diagonal block which contains the capacitance coefficients of the resonators within a single unit cell. We will use the notation $\Cf$ to denote the infinite matrix that contains all the $C_{ij}^m$ coefficients, for all $1\leq i,j\leq N$ and all $m\in\Lambda$.
\end{defn}

A final, important quantity for the analysis in this work is the truncated capacitance matrix $C_\trm$. This is obtained by keeping only $N|I_r|\times N|I_r|$ coefficients from $\Cf$, to give a matrix that is the same size as $C_\frm$. A schematic of the various pieces of notation used in this article and how they related to each other is given in \Cref{fig:notation}. The proof strategy deployed in this work is to compare the spectra of $C_\frm$ with that of $C_\trm$, and then let $r\to\infty$ in order to approximate the spectrum of $\Cf$. In particular, the modes that we will compare are \emph{defect modes}, which are spatially localized modes that exist due to the presence of defects in the otherwise periodic material, an example of which is shown in \Cref{fig:modeplot}.

\begin{figure}
\centering
\includegraphics[width=0.6\linewidth]{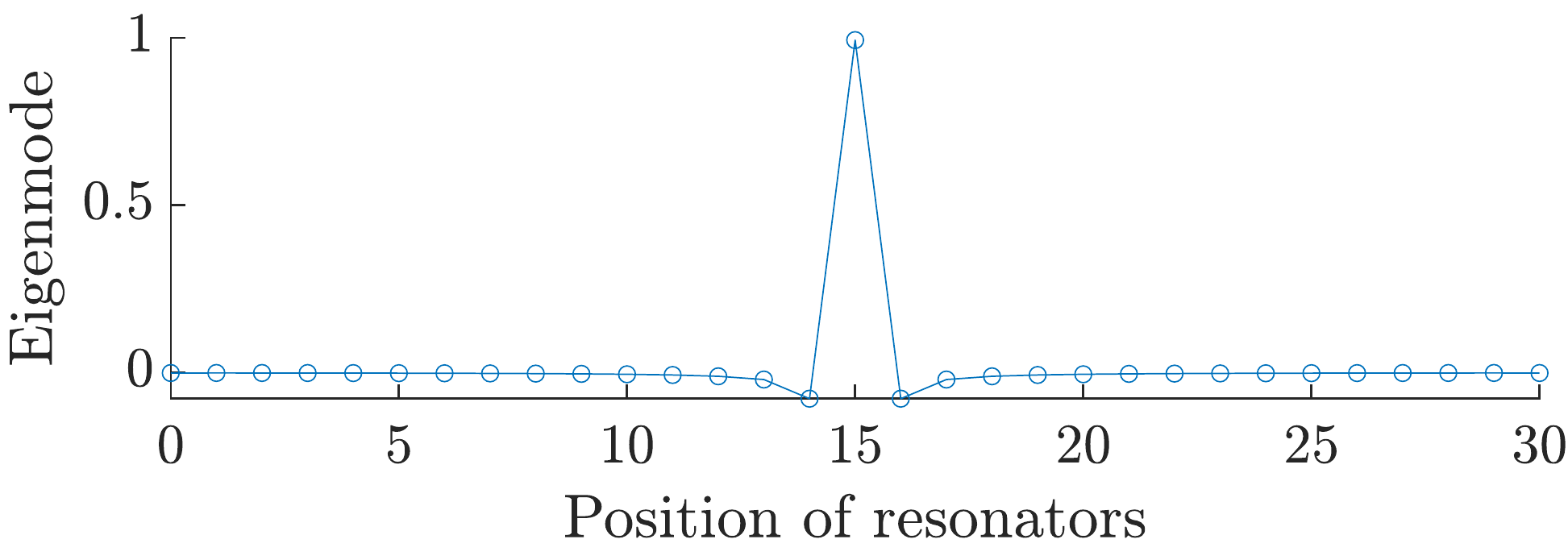}
\caption{An example of a localized defect mode for a system of 31 resonators. The eigenvalues of the finite matrix $B_\trm C_\frm$ are computed, where $C_\frm$ is the generalized capacitance matrix for a system of evenly spaces resonators and $B_\trm$ is the identity matrix but with the central entry $(B_\trm)_{11}^0=2$.} \label{fig:modeplot}
\end{figure}

We will model defect modes through pre-multiplication by a defect matrix $\Bf$. For each $m\in \Lambda$, we let $B^m$ be an $N\times N$ diagonal matrix
\begin{equation}
B^m  = \begin{pmatrix} b_1^m & 0 &\cdots & 0\\ 0&b_2^m&\cdots & 0 \\ \vdots & \vdots & \ddots & \vdots \\ 0&0& \cdots & b_N^m\end{pmatrix},
\end{equation}
where the diagonal entries $b_i^m$ are real-valued parameters. The coefficients $b_i^m$ describe perturbations of the material parameters (wave speed and contrast parameter) associated to each resonator $D_i^m$ (for further details we refer to \Cref{sec:asymptotics}). In this work, we only consider \emph{compact} defects, where $b_i^m=1$ for all but finitely many $i$ and $m$. For the infinite structure, we let $\Bf$ be the infinite block-diagonal matrix that contains $B^m$ for all $m\in \Lambda$. Under the assumption on the $b_i^m$, $\Bf$ is said to be a compact perturbation of the identity. The spectrum of the infinite structure is given by the solutions to the spectral problem
\begin{equation}\label{eq:eig}
\Bf\Cf \uf = \lambda \uf.
\end{equation}
For the finite structure of size $r$, we let $B_\trm$ be the block-diagonal matrix $(B^m), m\in I_r$ and consider the spectral problem
$$B_\trm C_\frm u = \lambda u.$$
An example of such a defect mode is shown in \Cref{fig:modeplot}. A system of 31 resonators is modelled, with the finite defect matrix $B_\trm$ chosen to be the identity, perturbed so that its central element is $(B_\trm)_{11}^0=2$.

The generalized capacitance matrix serves not only as a canonical model for coupled resonators (whose interaction terms decay as $r^{-1}$), but can also be derived from first principles in certain physical settings. For example, in \Cref{sec:asymptotics} we briefly explain how this model arises for a system of high-contrast resonators in which case the eigenstates of the generalized capacitance matrix fully characterize the subwavelength resonant spectrum of the system.

\subsection{Convergence of capacitance coefficients}
Based on the layer-potential characterization of capacitance, we prove in this section that the capacitance coefficients of a large but finite structure converge, as the size grows, to corresponding coefficients of the infinite structure. We begin with the following result, which collects some well-known results on the capacitance matrices \cite{ammari2021functional, diaz2011positivity}.
\begin{lem}
	Let $\widehat{C}^\alpha$ and $C_\frm$ be the quasi-periodic and finite capacitance matrix, respectively. Then 
	\begin{itemize}
		\item[(i)] $\widehat{C}^\alpha$ and $C_\frm$ are symmetric, positive definite matrices;
		\item[(ii)] $\widehat{C}^\alpha$ and $C_\frm$ are strictly diagonally dominant matrices;
		\item[(iii)] We have $(\widehat{C}^\alpha)_{ii}> 0$ and $(C^{mm}_\frm)_{ii} > 0$. Moreover, for $i\neq j$ and $m\neq n$ we have $(\widehat{C}^\alpha)_{ij}< 0$ and $(C^{mn}_\frm)_{ij} < 0$.
	\end{itemize}
\end{lem}

The next result shows that a fixed block of the infinite capacitance matrix is approximately equal to corresponding block of the capacitance matrix of the finite structure. In other words, the finite-structure capacitance coefficients can be approximated through the infinite structure as long as we are sufficiently far away from the edges of the finite structure.
\begin{thm} \label{thm:approx}
	For fixed $m,n \in \Lambda$, we have as $r\to \infty$, 
	$$\lim_{r\to \infty} C^{mn}_\frm(r) = C^{m-n}.$$
\end{thm}	
\begin{proof}
Firstly, observe that 
$$\S_{\D_\frm}[\psi] = \sum_{m\in I_r} \S_{D+m}[\psi_m],$$
where $\psi_m = \psi|_{\partial D+m}$.
Recall that the quasi-periodic single-layer potential is defined as 
$$\S_D^\alpha[\phi] = \int_{\p D} \sum_{m\in \Lambda} G(x-y-m)e^{\iu \alpha\cdot m} \phi(y) \dx \sigma.$$
Given $\phi\in L^2(D)$, we define $\phi_m^\alpha\in L^2(D+m)$ as
$$\phi_m^\alpha(y) = \phi(y-m)e^{\iu \alpha\cdot m}.$$
Then it is clear that 
$$\S_D^\alpha[\phi] = \sum_{m\in \Lambda} \S_{D+m}[\phi_m^\alpha].$$
We can then decompose 
\begin{align*}
	\S_D^\alpha[\phi] &= \sum_{m\in I_r} \S_{D+m}[\phi_m^\alpha] +  \int_{\p D} \sum_{m\in \Lambda\setminus I_r} G(x-y-m)e^{\iu \alpha\cdot m} \phi(y) \dx \sigma \\
	&= \S_{\D_\frm}[\phi^\alpha] + \Rc^\alpha[\phi],
\end{align*}
where, in the operator norm, $\Rc^\alpha = o(1)$ as $r\to \infty$. From the Neumann series, we now have
\begin{equation}\label{eq:main}
	(\S_D^\alpha)^{-1}[\chi_{\p D_i}] = \S_{\D_\frm}^{-1}[\chi_i^\alpha] + o(1),
\end{equation}
where  $\chi_i^\alpha $ is defined as
$$\chi_i^\alpha = \sum_{m\in I_r} \chi_{\p D_i^m}e^{\iu \alpha \cdot m}.$$	
From \Cref{lem:bound} in \Cref{sec:alp}, we know that the error term in \eqref{eq:main} holds uniformly in $\alpha$. If $m,n\in I_r$ are fixed and $i,j=1,...,N$, we then have from \eqref{eq:main} that
\begin{align*}
	C_{ij}^{m-n} &= \frac{1}{|Y^*|}\int_{Y^*}\int_{\p D_i}e^{-\iu\alpha\cdot (m-n)} \S_{\D_\frm}^{-1}[\chi_j^\alpha]\dx \sigma  \dx \alpha + o(1) \\
	&= \int_{\p D_i^m} \S_{\D_\frm}^{-1}[\chi_{\p D_j^n}]  \dx \sigma + o(1) \\
	&= (C^{mn}_\frm)_{ij}(r) + o(1).
\end{align*}
This proves the claim.
\end{proof}

\begin{figure}
	\begin{subfigure}{\linewidth}
	%\centering
	\includegraphics[width=0.55\linewidth]{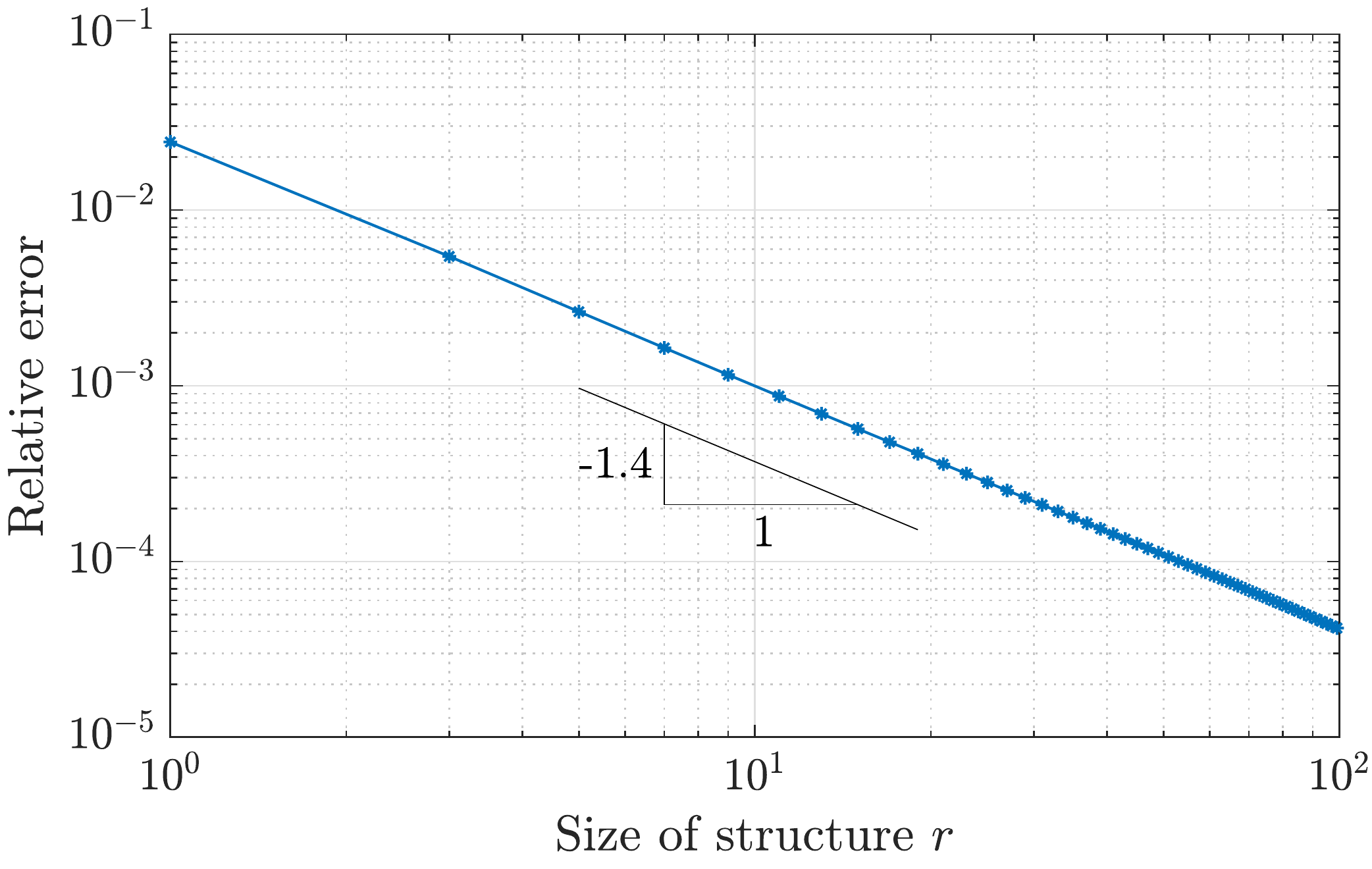}
	\begin{tikzpicture}
	\node at (0,0) {$\cdots$};
%	\draw[fill=white!80!gray] (0.5,0) circle(0.2);
%	\draw[fill=white!80!gray] (1.2,0) circle(0.2);
%	\draw[fill=white!80!gray] (1.9,0) circle(0.2);	
%	\draw[fill=white!80!gray] (2.6,0) circle(0.2);
%	\draw[fill=white!80!gray] (3.3,0) circle(0.2);
%	\draw[fill=white!80!gray] (4,0) circle(0.2);
%	\draw[fill=white!80!gray] (4.7,0) circle(0.2);
	\shade[ball color = gray!40] (0.5,0) circle (0.2);
	\draw (0.5,0) circle (0.2);
	\shade[ball color = gray!40] (1.2,0) circle (0.2);
	\draw (1.2,0) circle (0.2);
	\shade[ball color = gray!40] (1.9,0) circle (0.2);
	\draw (1.9,0) circle (0.2);
	\shade[ball color = gray!40] (2.6,0) circle (0.2);
	\draw (2.6,0) circle (0.2);
	\shade[ball color = gray!40] (3.3,0) circle (0.2);
	\draw (3.3,0) circle (0.2);
	\shade[ball color = gray!40] (4,0) circle (0.2);
	\draw (4,0) circle (0.2);
	\shade[ball color = gray!40] (4.7,0) circle (0.2);
	\draw (4.7,0) circle (0.2);
	\node at (5.25,0) {$\cdots$};
	\node[white] at (-0.5,-2.5) {.};
	\end{tikzpicture}
	\caption{One-dimensional lattice} \label{fig:err}
	\end{subfigure}

	\vspace{0.4cm}

	\begin{subfigure}{\linewidth}
	%	\centering
		\includegraphics[width=0.55\linewidth]{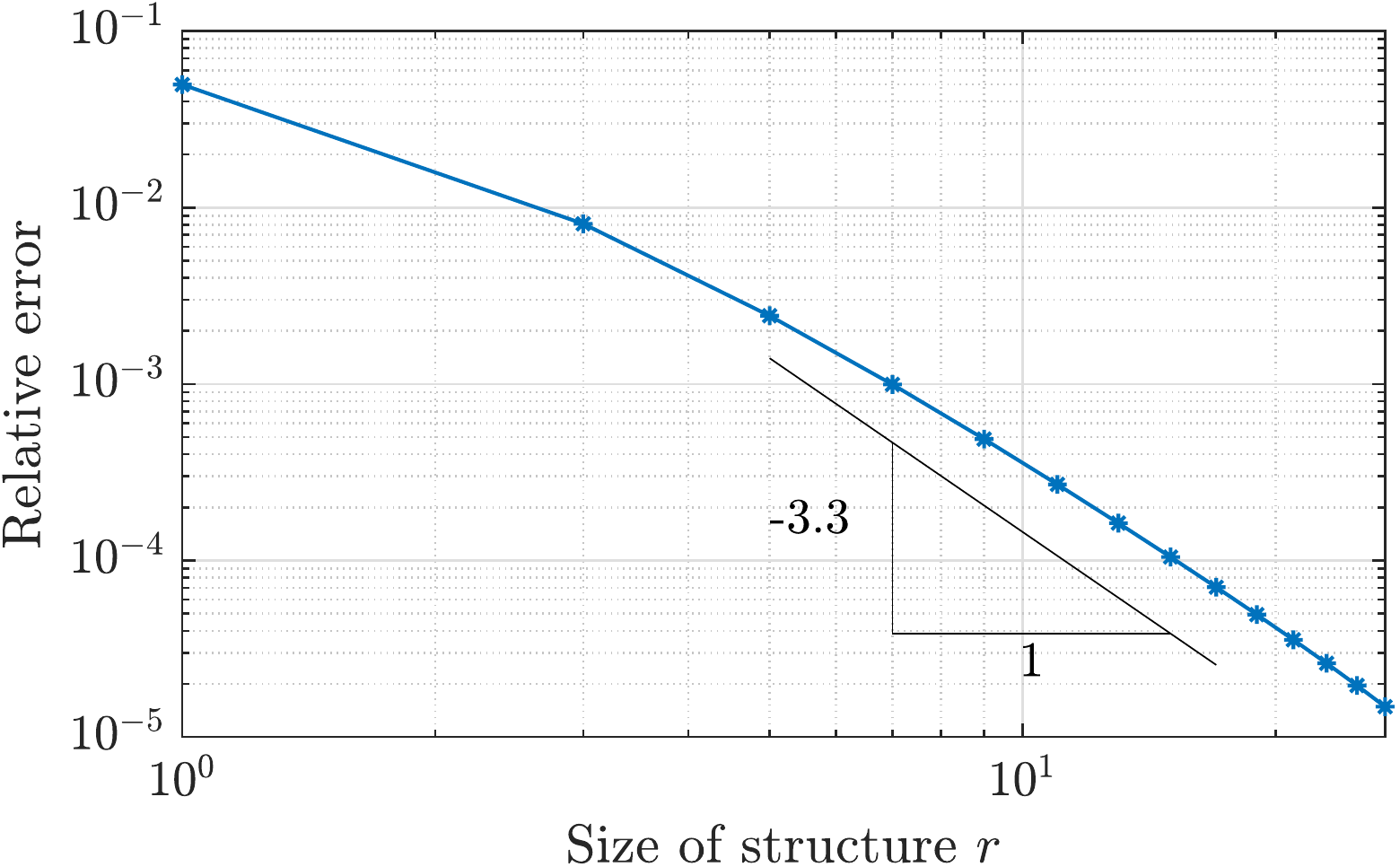}
	\begin{tikzpicture}%[scale=0.75]
		\foreach \x in {1,2,...,5}{
		\foreach \y in {1,2,...,5}{
%		\draw[fill=white!80!gray] (0.7*\x,0.7*\y) circle(0.2);
		\shade[ball color = gray!40] (0.7*\x,0.7*\y)  circle (0.2);
		\draw (0.7*\x,0.7*\y)  circle (0.2);
		\node at (0,0.7*\y) {$\cdots$};
		\node at (4.3,0.7*\y) {$\cdots$};
		\node at (0.7*\x,0.1) {$\vdots$};
		\node at (0.7*\x,4.3) {$\vdots$};
		}}
		\node[white] at (-1,-0.5) {.};
	\end{tikzpicture}
		\caption{Two-dimensional square lattice} \label{fig:errsq}
	\end{subfigure}
	
	\vspace{0.4cm}
	
	\begin{subfigure}{\linewidth}
	%\centering
	\includegraphics[width=0.55\linewidth]{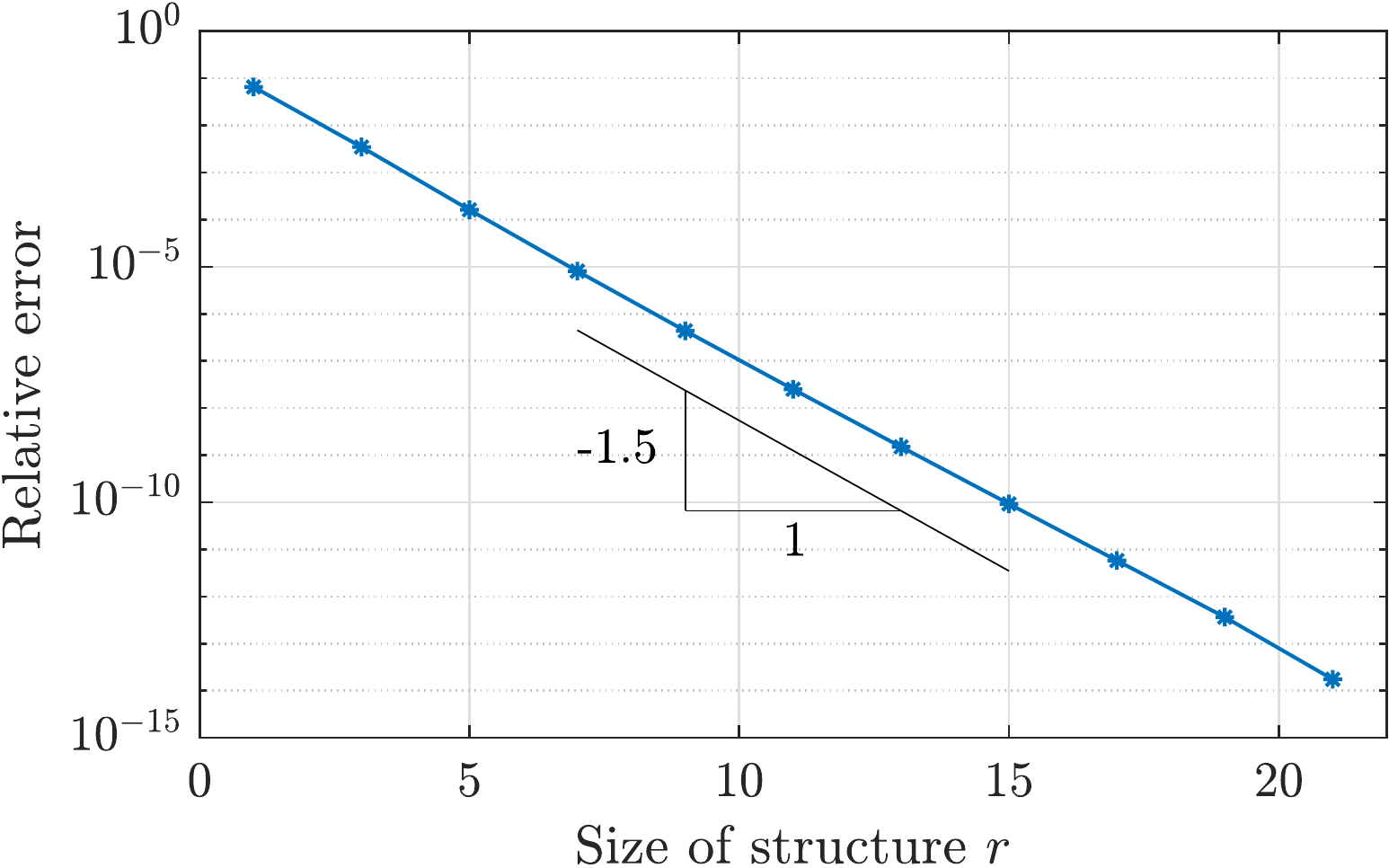}
	\begin{tikzpicture}
	
		\begin{scope}[xshift=0.15cm,yshift=0.15cm]
		\foreach \x in {1,2,...,5}{
		\foreach \y in {1,2,...,5}{
		\shade[ball color = gray!10] (0.7*\x,0.7*\y) circle (0.2);
		\draw (0.7*\x,0.7*\y) circle (0.2);
		}}
		\end{scope}
	
		\foreach \x in {1,2,...,5}{
		\foreach \y in {1,2,...,5}{
		\shade[ball color = gray!10] (0.7*\x,0.7*\y) circle (0.2);
		\draw (0.7*\x,0.7*\y) circle (0.2);
		\node at (0,0.7*\y) {$\cdots$};
		\node at (4.3,0.7*\y) {$\cdots$};
		\node at (0.7*\x,0.1) {$\vdots$};
		\node at (0.7*\x,4.3) {$\vdots$};
		}}
		
		\begin{scope}[xshift=-0.15cm,yshift=-0.15cm]
		\foreach \x in {1,2,...,5}{
		\foreach \y in {1,2,...,5}{
		\shade[ball color = gray!10] (0.7*\x,0.7*\y) circle (0.2);
		\draw (0.7*\x,0.7*\y) circle (0.2);
		}}
		\end{scope}
		
		\node[white] at (-1,-0.5) {.};
		%\node[white] at (5,-0.5) {.};
	\end{tikzpicture}
	\caption{Three-dimensional cubic lattice} \label{fig:err_cube}
	\end{subfigure}
	\caption{Convergence of the capacitance coefficient of large finite lattices. For each lattice, there is a single resonator in the unit cell ($N=1$). In each case, we plot $|(C_\frm)^0_{11}-C^0_{11}|$ for increasing size $r$ of the finite structure. Observe the log-log scales in (a) and (b) and the semi-log scale in (c), which correspond to algebraic and exponential convergence, respectively.} \label{fig:capconvergence}
\end{figure}

The numerical results presented in \Cref{fig:capconvergence} demonstrate the convergence of the capacitance coefficients, as established by \Cref{thm:approx}. We plot $|(C_\frm)^0_{11}-C^0_{11}|$ for a one-dimensional, two-dimensional, and three-dimensional lattice. When $d=1$ or $d=2$, the convergence is algebraic. The crucial property here is that the dimension of the lattice is less than the three dimensions of the underlying differential problem. As a result, waves are able to propagate away from the structure in the ``spare'' dimensions. This introduces long-range interactions to the system, as non-adjacent resonators are able to interact by coupling with the far field. Conversely, when $d=3$ the dimension of the lattice is maximal (in the sense that we have a three-dimensional lattice in three-dimensional space). In this case, we see exponential convergence as there are no ``spare'' dimensions that allow waves to propagate away from the structure and couple with the far field.

\section{Convergence to pure point spectrum} \label{sec:purep}
In this section, we study a problem where the infinite structure has a pure point spectrum, corresponding to a localized mode. We introduce a defect to the model in order to create such a mode. For a finite, truncated structure, there will be an eigenvalue arbitrarily close to the pure point spectrum.

\subsection{Example of a defect structure}\label{sec:defect}
Before developing any convergence theory, we present an example of a defect structure exhibiting a pure point spectrum, corresponding to a localized mode. We take a lattice with a single resonator $N = 1$ inside each unit cell. We take a single resonator with perturbed (``defect'') material parameter. In other words,
\begin{equation}
	b_1^m = \begin{cases} 1, \ &m\neq 0, \\ 1 + \eta, \ &m=0, \end{cases}
\end{equation}
for some parameter $\eta > -1$. Observe that $\eta=0$ corresponds to the unperturbed case, and $|\eta|$ describes the magnitude of the perturbation. The eigenvalues of the (infinite-dimensional) generalized capacitance matrix $\Bf\Cf$ in this setting was studied in \cite{anderson}. It was found that $\lambda$ is an eigenvalue of $\Bf\Cf$ if and only if it is a root of the equation
\begin{equation} \label{eq:M1}
	\frac{\eta}{|Y^*|}\int_{Y^*} \frac{\lambda_1^\alpha }{\lambda-\lambda_1^\alpha} \dx \alpha = 1,
\end{equation}
where $\lambda_1^\alpha$ is the single eigenvalue of the quasi-periodic capacitance matrix $\widehat{C}^\alpha$ of the unperturbed periodic structure. This equation has a solution $\lambda=\lambda_0$ precisely in the case $\eta>0$. In other words, the defect induces an eigenvalue $\lambda_0$ in the pure point spectrum of $\Bf\Cf$, corresponding to an exponentially localized eigenmode. An example of such a localized eigenmode was shown in \Cref{fig:modeplot}.

\subsection{Convergence of defect modes}
In this section, we prove that, if the infinite structure has a localized mode, there will be an eigenvalue of the truncated structure arbitrarily close to the localized frequency.

We let $\Cf$ denote the infinite capacitance matrix. As before, we let $C_\frm$ denote the capacitance matrix of a finite structure of size $N|I_r|\times N|I_r|$. Furthermore, we let $C_\trm$ denote the truncated matrix of $\Cf$ of size  $N|I_r|\times N|I_r|$, and similarly let $B_\trm$ be the truncation of $\Bf$. At this point, we emphasize that $C_\trm$ is ``nonphysical'' in the sense that it does not correspond to a capacitance matrix associated to any physical structure but, rather, to the finite matrix obtained by simply truncating the infinite matrix $\Cf$.

We assume that $\Bf\Cf$ has a localized eigenmode $\uf$, and let $u_\trm$ be the truncation of $\uf$ of size $N|I_r|$. The first result follows only from the decay of the localized mode.
\begin{lem}\label{lem:ct}
	Assume that $\Bf$ is a compact perturbation of the identity, such that $\Bf\Cf$ has a localized eigenmode $\uf$ with corresponding eigenvalue $\lambda$. Then there is an eigenvalue $\tilde\lambda = \tilde\lambda(r)$ of $ B_\trm C_\trm$ satisfying
	$$\lim_{r\to \infty}\tilde\lambda(r) = \lambda.$$
\end{lem}
\begin{proof}
	We let $\uf_\trm$ be the infinite vector obtained by padding $u_\trm$ with $0$. Since $\uf$ is in $\ell^2(\Lambda)$, for any $\epsilon > 0$ we can choose large enough $r$ so that 
	$$\|\uf-\uf_\trm\|_{\ell^2} < \epsilon.$$
	Since $\Bf\Cf$ is a bounded operator, we then have
	$$\|\lambda \uf -\Bf\Cf \uf_\trm\|_{\ell^2} < K\epsilon,$$
	for some $K>0$. Restricting to the finite block of size $r$, we have
	$$\|\lambda u_\trm -  B_\trm C_\trm u_\trm\|_2 < K\epsilon.$$
	In other words, $\lambda$ is in the $K\epsilon$-pseudospectrum of $ B_\trm C_\trm$, and since $ B_\trm C_\trm$ is normal, we have an eigenvalue $\tilde \lambda$ of $ B_\trm C_\trm$ satisfying 
	$$|\tilde \lambda(r) - \lambda| = K\epsilon.$$
	This proves the claim.
\end{proof}

Next, we study the properties of $C_\frm$ as the size of the finite structure increases.
\begin{lem}\label{lem:dirichlet}
	For $i=1,...,N+1$, assume that $B_i \subset\R^3$ are disjoint, connected domains and let
	$$B=\bigcup_{n=1}^N B_i \quad \widetilde B = \bigcup_{n=1}^{N+1} B_i.$$
	Let $C_{ij}$, $\widetilde C_{ij}$ denote the capacitance coefficients associated to $B$ and $\widetilde{B}$, respectively. Then
	$$C_{ii} \leq \widetilde{C}_{ii} \quad i=1,...,N.$$
\end{lem}
\begin{proof}
	We will use a variational characterization of the capacitance coefficients. Let $\Hc = \{v\in H^1_\text{loc}(\R^3) \mid v(x) \sim |x|^{-1} \text{ as } x\to \infty \}$ and let
	\begin{align*}
		\V &= \{v\in \Hc \mid v|_{\p B_j} = \delta_{ij} \text{ for } j=1,...,N\},\\
		\widetilde\V &= \{v\in \Hc \mid v|_{\p B_j} = \delta_{ij} \text{ for } j=1,...,N+1\}.
	\end{align*}
Observe that  $\widetilde\V \subset \V$. It then follows that
$$C_{ii} = \min_{v\in \V} \int_{\R^3}|\nabla v|^2\dx x \leq  \min_{v\in \widetilde\V} \int_{\R^3}|\nabla v|^2\dx x = \widetilde {C}_{ii}.$$ 
\end{proof}
\begin{rmk}
	\Cref{lem:dirichlet} states that the diagonal capacitance coefficients will always increase when adding additional resonators. In the physical situation of electrostatics this result is intuitive: the self-capacitance of a conductor can only increase if additional conductors are introduced.
\end{rmk}

\begin{lem}\label{lem:K}
	As $r\to \infty$, we have $\|C_\frm\|_2 < K$ for some $K$ independent of $r$. 
\end{lem}
\begin{proof}
	We know that the capacitance matrix $C_\frm$ is diagonally dominant:
	$$(C^{mm}_\frm)_{ii} > \sum_{n\in \Z, j\neq i} \bigl|(C^{mn}_\frm)_{ij}\bigr|,$$
	for any $i,m$. For fixed $i$ and $m$, we know from \Cref{lem:dirichlet} that $(C^{mm}_\frm)_{ii}(r)$ is increasing in $r$, and for all $r$ we have
	$$(C^{mm}_\frm)_{ii}(r) < C^{0}_{ii},$$
	where, as before, $C_{ii}^0$ is the corresponding entry of the infinite capacitance matrix $\Cf$. In particular, the eigenvalues of $C_\frm(r)$ are bounded as $r\to \infty$, which shows the claim.
\end{proof}

As discussed above, the matrix $C_\trm$ appearing in \Cref{lem:ct} is nonphysical, as it is a truncation of the matrix for the infinite system. Instead, we need to phrase the result for the matrix $C_\frm$, which describes the finite system. The following theorem is the main result of this section.
\begin{thm}\label{thm:pp}
	Assume that $\Bf$ is a compact perturbation of the identity, such that $\Bf\Cf$ has a localized eigenmode $\uf$ with corresponding eigenvalue $\lambda$. Then there is an eigenvalue $\hat\lambda = \hat\lambda(r)$ of $ B_\trm C_\frm$ satisfying
	$$\lim_{r\to \infty}\hat\lambda(r) = \lambda.$$	
\end{thm}
\begin{proof}
	We let
	$$ K_1 = \sup_{r>0} \|C_\frm(r) - C_\trm\|_2,$$
	and observe from \Cref{lem:K} that $K < \infty$. We also let
	$$K_2 = \|B_{\frm}\|_2.$$
	Given $\epsilon >0$, we pick $r_0>0$ such that the following four terms are small:
	$$\|C_{0,\frm} - C_{0,\trm}\|_2 < \frac{\epsilon}{4K_2}, \quad \|u_\trm - u_{0,\trm}\|_2 < \frac{\epsilon}{4K_1K_2}, \quad \|B_{\frm}(C_\trm - C_{0,\trm})u_{0,\trm}\|_2 < \frac{\epsilon}{4}, \quad \|B_{\frm}(C_\frm - C_{0,\frm})u_{0,\trm}\|_2 < \frac{\epsilon}{4}$$
	for all $r$ large enough; the first inequality follows from \Cref{thm:approx} while the subsequent inequalities follow from the $\ell^2(\Lambda)$-decay of $u$. Here, $C_{0,\trm}, u_{0,\trm},$ and $C_{0,\frm}$ are the truncations of  $C_{\trm}, u_{\trm},$ and $C_{\frm}$ to the smaller lattice of radius $r_0$ (padded with zero where needed for the matrix operations).	We know from  \Cref{lem:ct} that we can take $r$ large enough so that $ B_\trm C_\trm$ has an eigenvalue $\tilde\lambda$ of distance  $\epsilon$ from $\lambda$. We then have
	\begin{multline}
		B_{\frm}C_\frm u_\trm = B_{\frm}C_\trm u_\trm + B_{\frm}(C_\frm - C_\trm)(u_\trm - u_{0,\trm}) + B_{\frm}(C_{0,\frm} - C_{0,\trm})u_{0,\trm} \\ + B_{\frm}(C_\frm - C_{0,\frm})u_{0,\trm} - B_{\frm}(C_\trm - C_{0,\trm})u_{0,\trm}.
	\end{multline}
	Then
	$$\|( B_\trm C_\frm - B_\trm C_\trm) u_\trm \|_2 < \epsilon,$$
	which means that there is an eigenvalue $\hat \lambda$ of distance $\epsilon$ from $\widetilde \lambda$, and hence $|\hat \lambda - \lambda| < 2\epsilon$.
\end{proof}

\begin{rmk}
	As an example, $\Bf$ and $\Cf$ as given in \Cref{sec:defect} satisfy the assumptions of \Cref{thm:pp}.
\end{rmk}

\subsection{Numerical illustration}

\begin{figure}
	\begin{subfigure}{\linewidth}
	%\centering
	\includegraphics[width=0.55\linewidth]{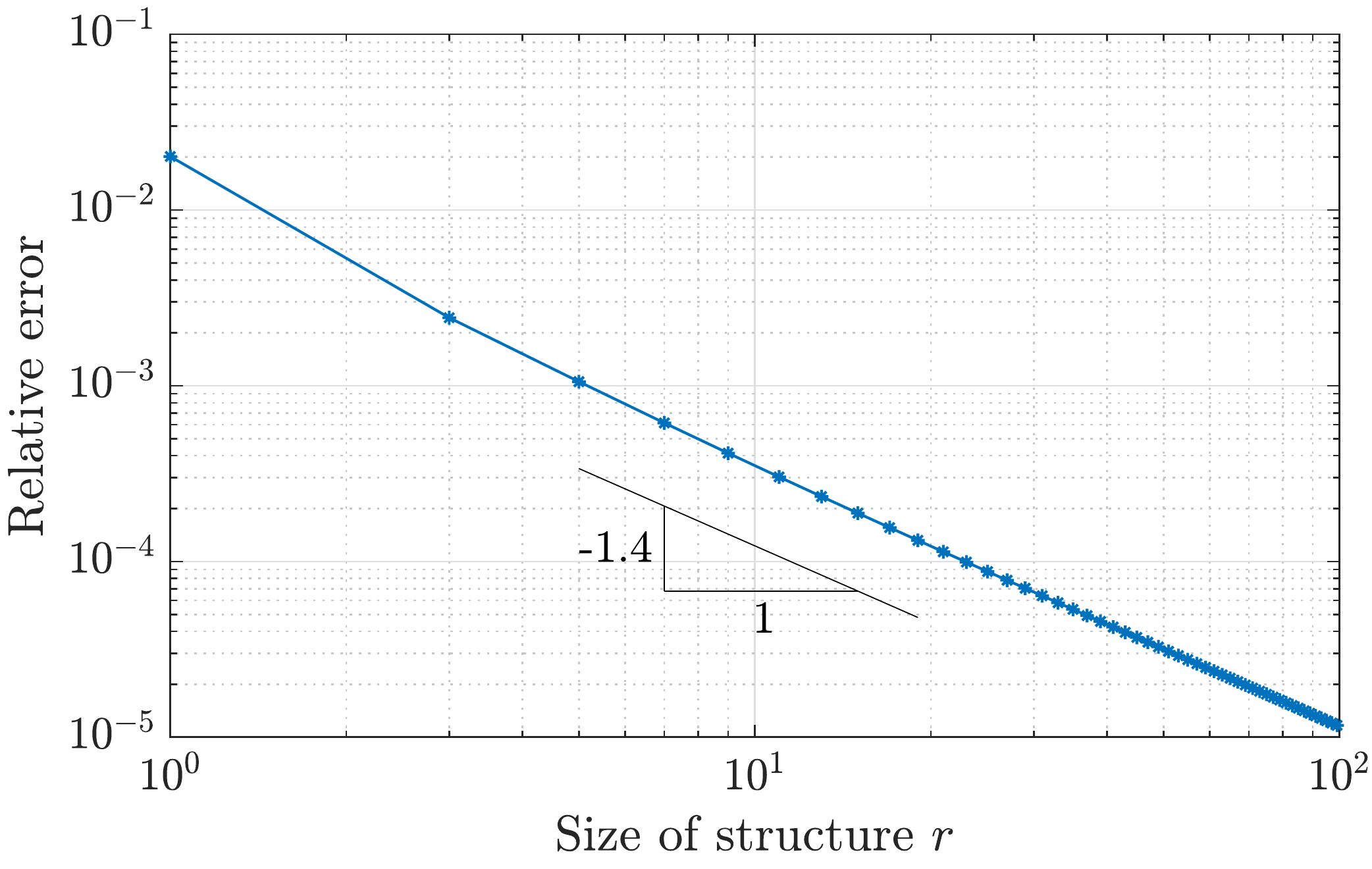}
	\begin{minipage}[b]{0.425\textwidth}
	\begin{tikzpicture}
	\node at (-0.2,0) {$\cdots$};
	\shade[ball color = gray!10] (0.5,0)  circle (0.2);
	\draw (0.5,0)  circle (0.2);
%	\draw[fill=white!80!gray] (0.5,0) circle(0.2);
	\node at (0.5,0.4) {\scriptsize $1$};
	\shade[ball color = gray!10] (1.2,0)  circle (0.2);
	\draw (1.2,0)  circle (0.2);
%	\draw[fill=white!80!gray] (1.2,0) circle(0.2);
	\node at (1.2,0.4) {\scriptsize $1$};
	\shade[ball color = gray!10] (1.9,0)  circle (0.2);
	\draw (1.9,0)  circle (0.2);
%	\draw[fill=white!80!gray] (1.9,0) circle(0.2);	
	\node at (1.9,0.4) {\scriptsize $1$};
	\shade[ball color = blue!90] (2.6,0)  circle (0.2);
	\draw (2.6,0)  circle (0.2);
%	\draw[fill=white!40!gray] (2.6,0) circle(0.2);
	\node at (2.6,0.4) {\scriptsize $1{+}\eta$};
	\shade[ball color = gray!10] (3.3,0)  circle (0.2);
	\draw (3.3,0)  circle (0.2);
%	\draw[fill=white!80!gray] (3.3,0) circle(0.2);
	\node at (3.3,0.4) {\scriptsize $1$};
	\shade[ball color = gray!10] (4,0)  circle (0.2);
	\draw (4,0)  circle (0.2);
%	\draw[fill=white!80!gray] (4,0) circle(0.2);
	\node at (4,0.4) {\scriptsize $1$};
	\shade[ball color = gray!10] (4.7,0)  circle (0.2);
	\draw (4.7,0)  circle (0.2);
%	\draw[fill=white!80!gray] (4.7,0) circle(0.2);
	\node at (4.7,0.4) {\scriptsize $1$};
	\node at (5.5,0) {$\cdots$};
%	\node[white] at (0,-2.5) {.};
	\end{tikzpicture}
	
	\vspace{0.6cm}
	
	\includegraphics[width=\linewidth]{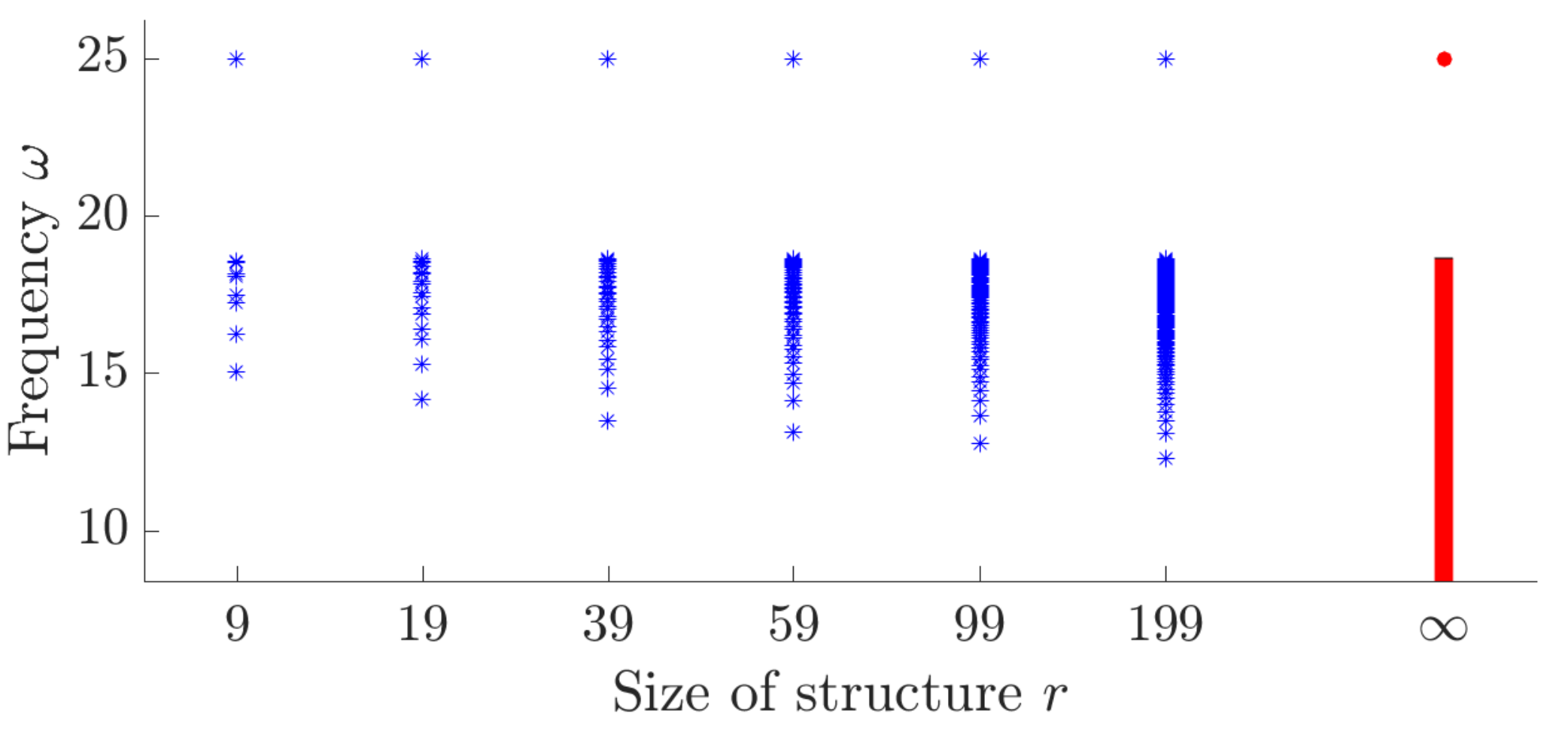}
	\end{minipage}
	\caption{One-dimensional lattice} \label{fig:defect}
	\end{subfigure}
	
	\vspace{0.2cm}

	\begin{subfigure}{\linewidth}
	%\centering
	\includegraphics[width=0.55\linewidth]{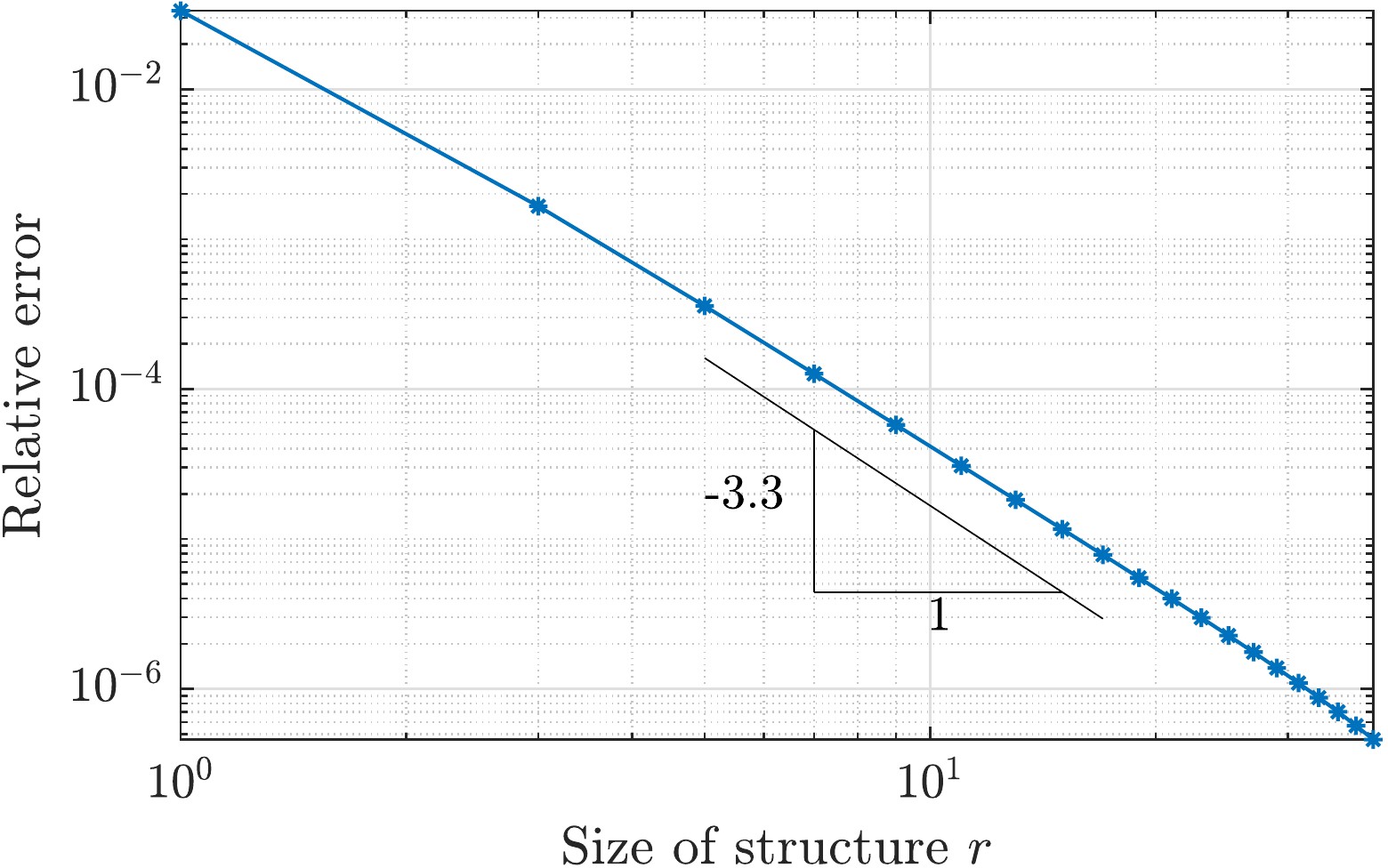}
	\begin{tikzpicture}%[scale=0.75]
		\foreach \x in {1,2,...,5}{
		\foreach \y in {1,2,...,5}{
%		\draw[fill=white!80!gray] (0.7*\x,0.7*\y) circle(0.2);
		\shade[ball color = gray!10] (0.7*\x,0.7*\y)  circle (0.2);
		\draw (0.7*\x,0.7*\y)  circle (0.2);
		\node at (0,0.7*\y) {$\cdots$};
		\node at (4.3,0.7*\y) {$\cdots$};
		\node at (0.7*\x,0.1) {$\vdots$};
		\node at (0.7*\x,4.3) {$\vdots$};
		}}
%		\draw[fill=white!40!gray] (0.7*3,0.7*3) circle(0.2);
%		\shade[ball color = gray!90] (0.7*3,0.7*3)  circle (0.2);
		\shade[ball color = blue!90] (0.7*3,0.7*3)  circle (0.2);
		\draw (0.7*3,0.7*3)  circle (0.2);
		\node[white] at (-1,-0.5) {.};
		%\node[white] at (5.8,-0.5) {.};
	\end{tikzpicture}
	\caption{Two-dimensional square lattice} \label{fig:defect_square}
	\end{subfigure}
	
	\vspace{0.2cm}
	
	\begin{subfigure}{\linewidth}
	\includegraphics[width=0.55\linewidth]{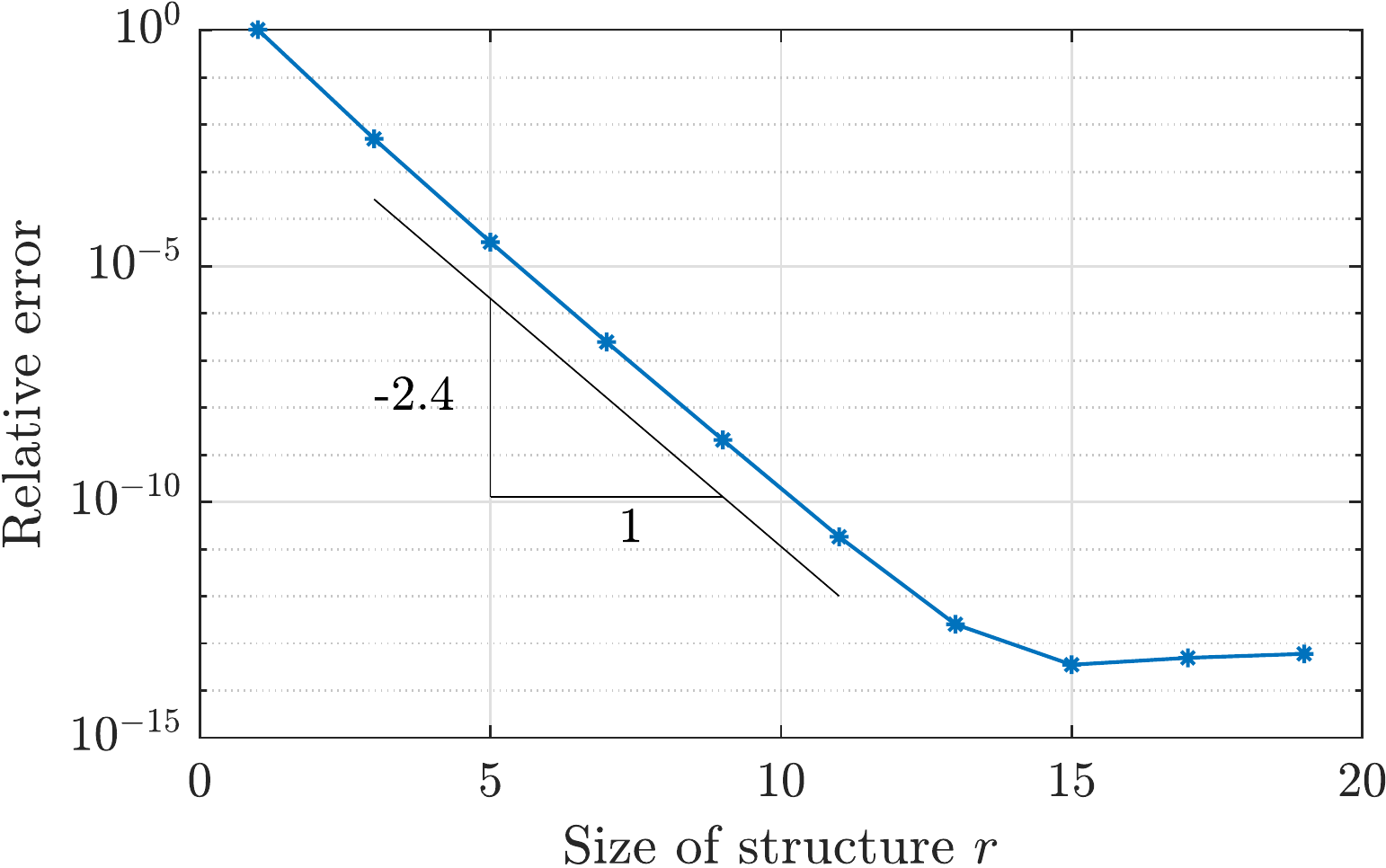}
	\begin{tikzpicture}
	
		\begin{scope}[xshift=0.15cm,yshift=0.15cm]
		\foreach \x in {1,2,...,5}{
		\foreach \y in {1,2,...,5}{
		\shade[ball color = gray!10] (0.7*\x,0.7*\y) circle (0.2);
		\draw (0.7*\x,0.7*\y) circle (0.2);
		}}
		\end{scope}
	
		\foreach \x in {1,2,...,5}{
		\foreach \y in {1,2,...,5}{
		\shade[ball color = gray!10] (0.7*\x,0.7*\y) circle (0.2);
		\draw (0.7*\x,0.7*\y) circle (0.2);
		\node at (0,0.7*\y) {$\cdots$};
		\node at (4.3,0.7*\y) {$\cdots$};
		\node at (0.7*\x,0.1) {$\vdots$};
		\node at (0.7*\x,4.3) {$\vdots$};
		}}
		\shade[ball color = blue!90] (0.7*3,0.7*3) circle (0.2);
		\draw (0.7*3,0.7*3) circle (0.2);
		
		\begin{scope}[xshift=-0.15cm,yshift=-0.15cm]
		\foreach \x in {1,2,...,5}{
		\foreach \y in {1,2,...,5}{
		\shade[ball color = gray!10] (0.7*\x,0.7*\y) circle (0.2);
		\draw (0.7*\x,0.7*\y) circle (0.2);
		}}
		\end{scope}
		
		\node[white] at (-1,-0.5) {.};
		%\node[white] at (5,-0.5) {.};
	\end{tikzpicture}
	\caption{Three-dimensional cubic lattice} \label{fig:defect_3d}
	\end{subfigure}
	\caption{Convergence of the frequency of the defect modes, for a defect on the central resonator (with $\eta=1$) created by perturbing a single entry of $\Bf$. (a) A one-dimensional lattice with a single resonator in the unit cell ($N=1$). The difference between the defect frequency computed for a finite structure and for the corresponding infinite structure scales as $O(r^{-1.4})$, where $r$ is the length of the truncated structure. The lower right plot shows the spectrum of successively larger lattices. In the geometry sketch on the right, the corresponding entry $b_1^m$ from the matrix $\Bf$ is shown above each resonator. (b) A two-dimensional square lattice with a single resonator in the unit cell ($N=1$). Here, the error scales as $O(r^{-3.3})$, where $r$ is the width of the (square) truncated structure. (c) A three-dimensional cubic lattice with a single resonator in the unit cell ($N=1$). Here, the error scales as $O(e^{-2.4r})$, where $r$ is the width of the (cubic) truncated structure.} \label{fig:point_defect_convergence}
\end{figure}
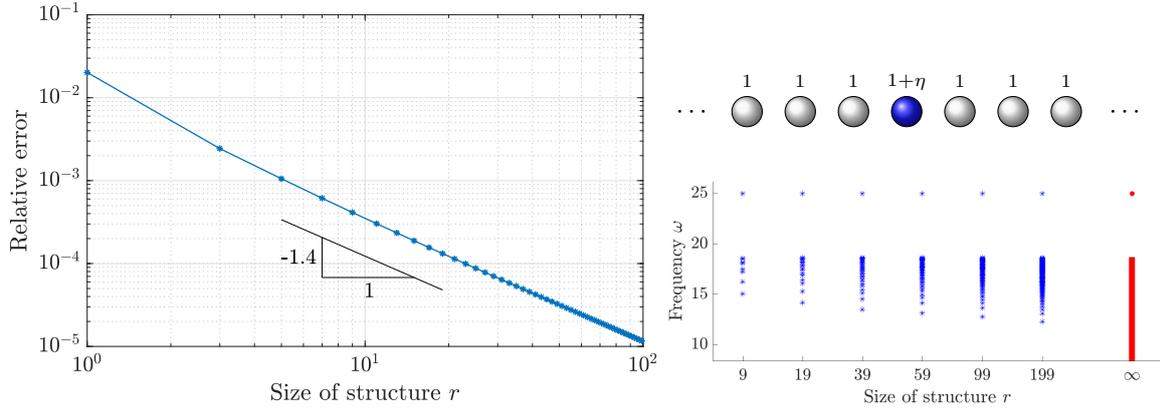
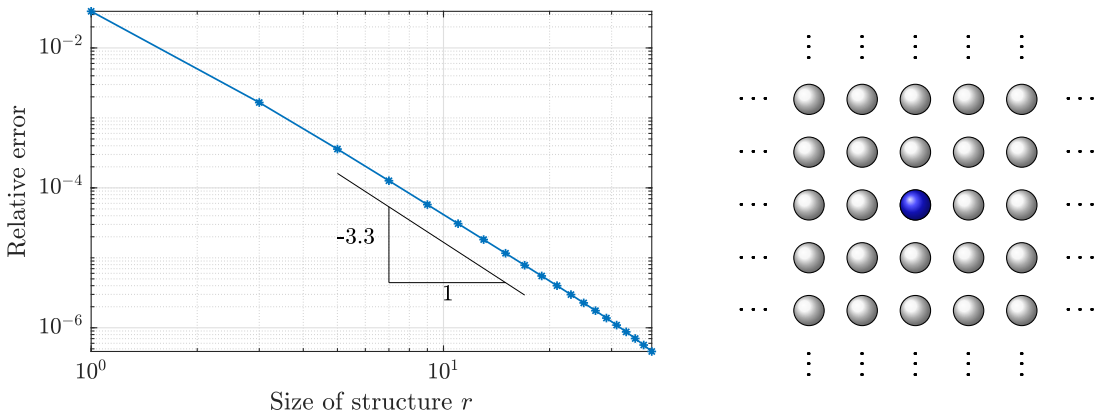
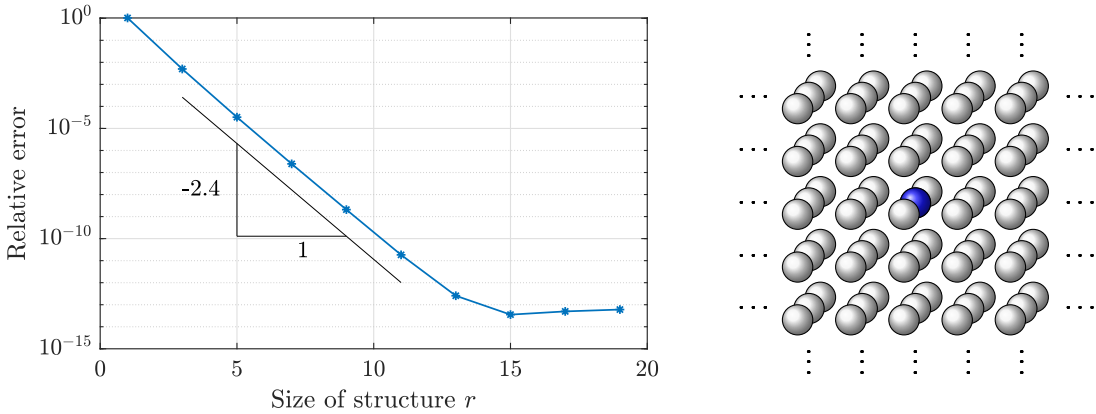

\Cref{fig:point_defect_convergence} shows the convergence of the difference between the defect frequency computed for a finite structure and for the corresponding infinite structure, computed analytically using \emph{e.g.} \eqref{eq:M1}. Comparing \Cref{fig:point_defect_convergence} with \Cref{fig:capconvergence}, it appears that the error of the frequency of the defect mode is inheriting the convergence rate of the capacitance coefficients. In other words, when $d=1$ or $d=2$, there are long-range interactions through coupling with the far-field, leading to algebraic convergence. In $d=3$, there are no ``spare'' dimensions and the convergence is exponential. This is consistent with the results for one-dimensional models \cite{lin2013resonances, lin2016perturbation, lu2022defect} and for two-dimensional lattices in two-dimensional problems \cite{lin2015twodim}.

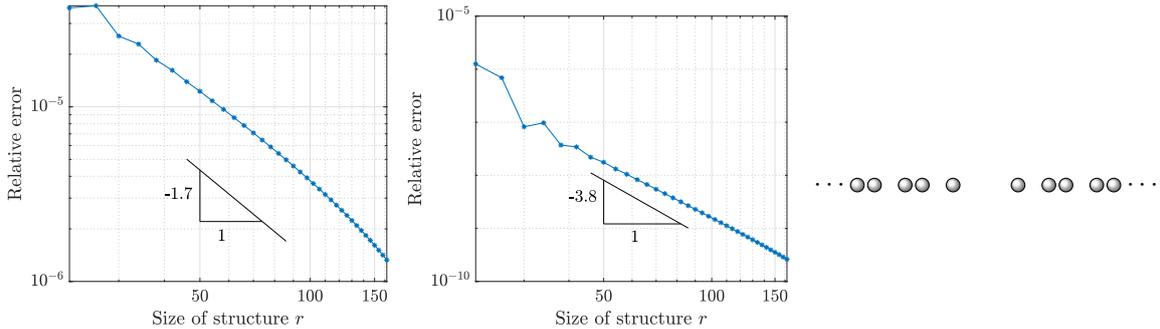
\begin{figure}
	\centering
	\includegraphics[width=0.33\linewidth]{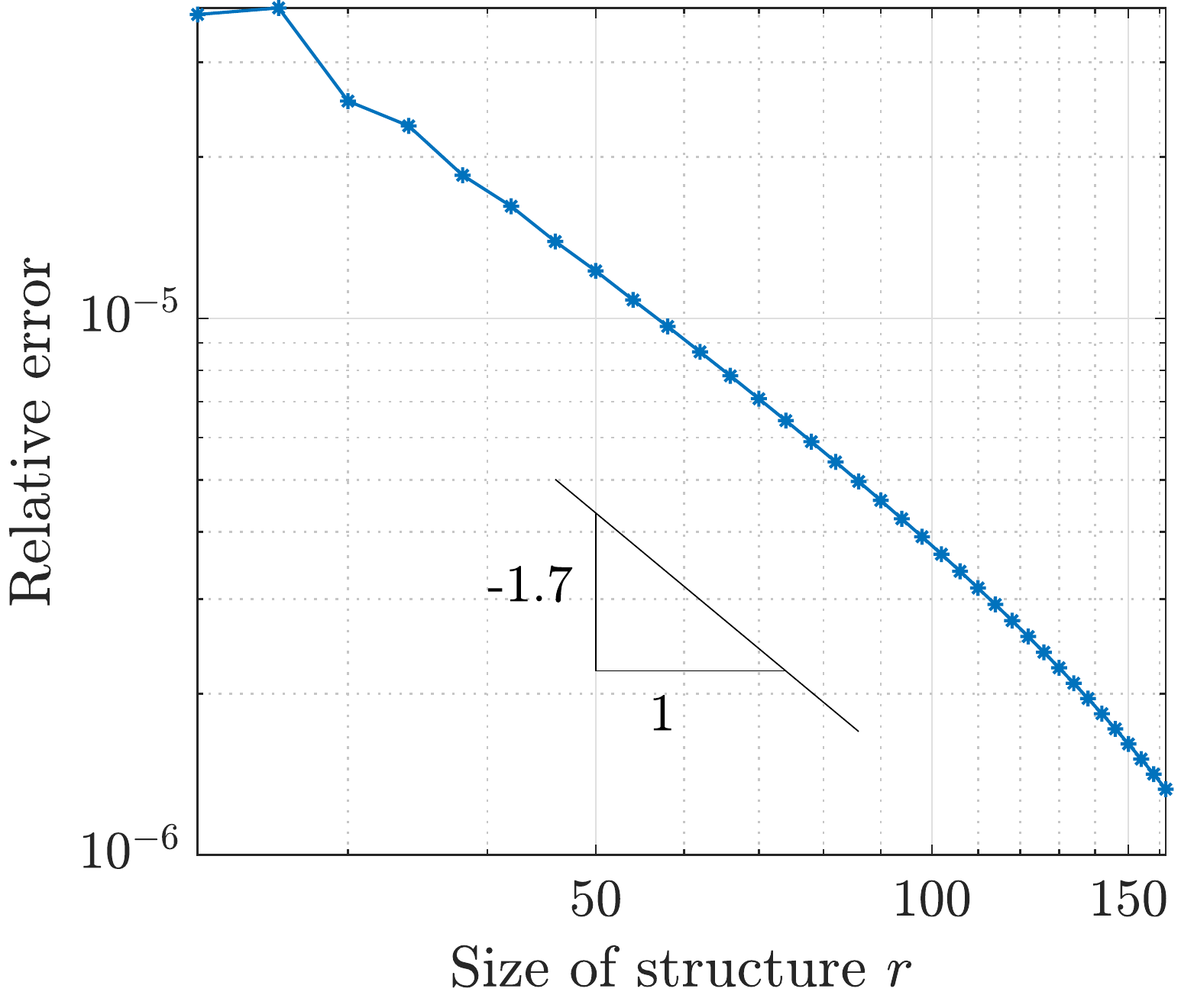}
	\includegraphics[width=0.33\linewidth]{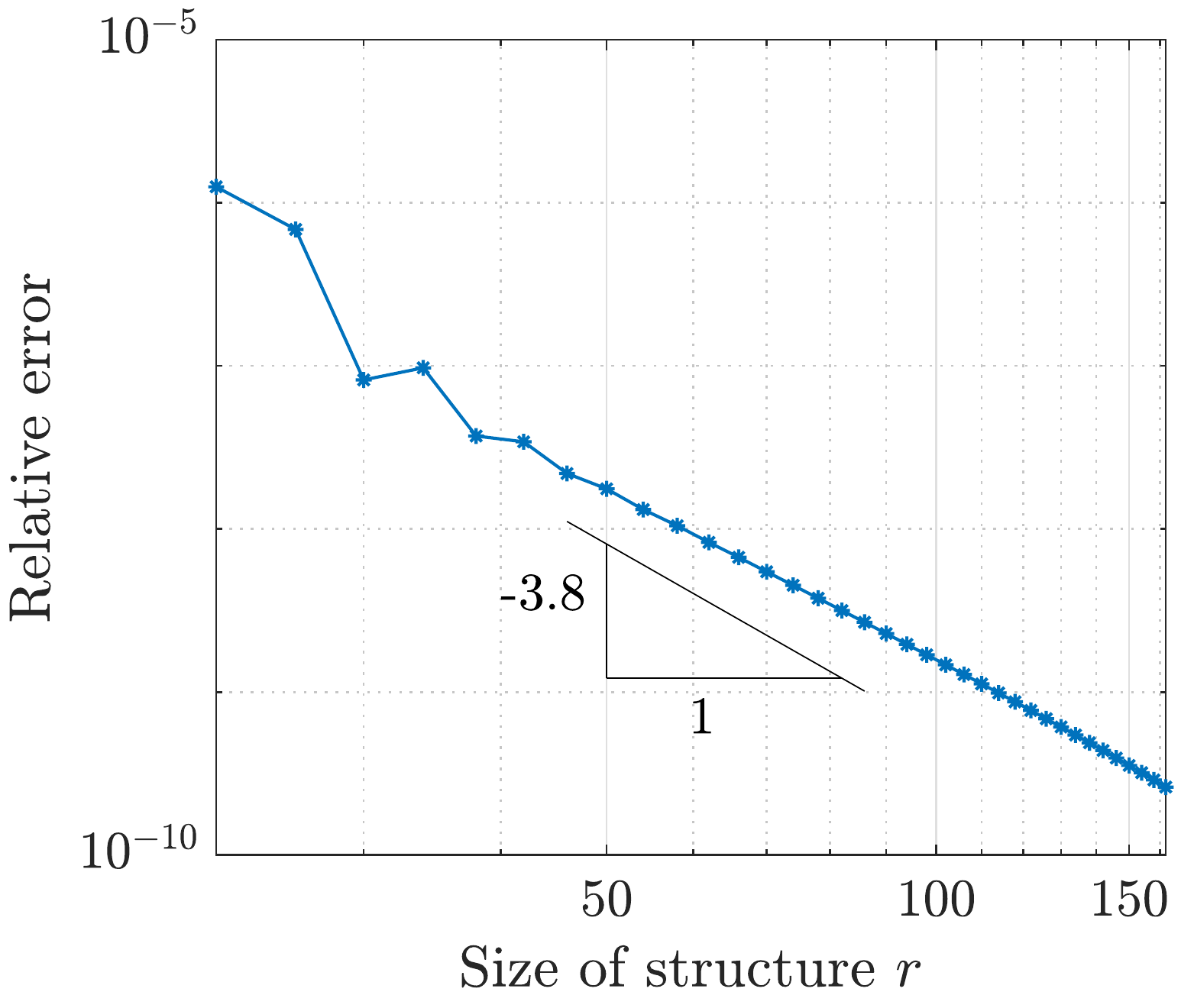}
	\begin{tikzpicture}[scale=0.45]
	\node[scale=1] at (-0.15,0) {$\cdots$};
%		\shade[ball color = gray!10] (0.7*\x,0.7*\y) circle (0.2);
%		\draw (0.7*\x,0.7*\y) circle (0.2);
		\shade[ball color = gray!10] (0.6,0) circle (0.2);
		\draw (0.6,0) circle (0.2);
		\shade[ball color = gray!10] (1.1,0) circle (0.2);
		\draw (1.1,0) circle (0.2);
		\shade[ball color = gray!10] (2,0) circle (0.2);
		\draw (2,0) circle (0.2);
		\shade[ball color = gray!10] (2.5,0) circle (0.2);
		\draw (2.5,0) circle (0.2);
		\shade[ball color = gray!10] (3.4,0) circle (0.2);
		\draw (3.4,0) circle (0.2);
		\shade[ball color = gray!10] (5.3,0) circle (0.2);
		\draw (5.3,0) circle (0.2);
		\shade[ball color = gray!10] (6.2,0) circle (0.2);
		\draw (6.2,0) circle (0.2);
		\shade[ball color = gray!10] (6.7,0) circle (0.2);
		\draw (6.7,0) circle (0.2);
		\shade[ball color = gray!10] (7.6,0) circle (0.2);
		\draw (7.6,0) circle (0.2);
		\shade[ball color = gray!10] (8.1,0) circle (0.2);
		\draw (8.1,0) circle (0.2);
		
%		\draw[fill=white!80!gray] (0.6,0) circle(0.2);
%		\draw[fill=white!80!gray] (1.1,0) circle(0.2);
%		\draw[fill=white!80!gray] (2,0) circle(0.2);
%		\draw[fill=white!80!gray] (2.5,0) circle(0.2);
%		\draw[fill=white!80!gray] (3.4,0) circle(0.2);
%		\draw[fill=white!80!gray] (3.9,0) circle(0.2);
%		\draw[fill=white!80!gray] (4.8,0) circle(0.2);
%		\draw[fill=white!80!gray] (5.3,0) circle(0.2);
%		\draw[fill=white!80!gray] (6.2,0) circle(0.2);
%		\draw[fill=white!80!gray] (6.7,0) circle(0.2);
%		\draw[fill=white!80!gray] (7.6,0) circle(0.2);
%		\draw[fill=white!80!gray] (8.1,0) circle(0.2);
		\node at (9,0) {$\cdots$};
		\node[white] at (0,-4) {.};
	\end{tikzpicture}
	\caption{Convergence of the frequency of the defect modes in a lattice with resonators arranged in pairs ($N=2$) and a defect corresponding to the two central resonators being removed. This gives two topologically protected edge modes. Here, the error scales as $O(r^{-1.7})$ for the even mode and $O(r^{-3.8})$ for the odd mode, where $r$ is the length of the truncated structure.} \label{fig:SSHdislocated}
\end{figure}

\begin{rmk}
Comparing \Cref{fig:capconvergence} and \Cref{fig:point_defect_convergence}, it appears that the error of the frequency of the defect mode is inheriting the convergence rate of the capacitance coefficients. While this is unsurprising, it turns out not to be the case for other types of defect. For example, in \Cref{fig:SSHdislocated} we show the convergence of the defect modes in a dislocated Su-Schrieffer-Heeger (SSH) lattice, which is a one-dimensional lattice of resonators arranged in pairs (so $N=2$). This system supports two defect modes that are known to be \emph{topologically protected} and benefit from enhanced robustness properties (see \cite{ammari2022robust} for details). The even mode experiences $O(r^{-1.7})$ convergence while the odd mode converges at a faster $O(r^{-3.8})$ rate. Understanding these different convergence rates is a valuable question for future study.
\end{rmk}

\section{Convergence to continuous spectrum} \label{sec:cont}
Through numerical illustrations, we can illustrate how the discrete spectrum of the truncated structure approximates the Floquet-Bloch spectral bands of the infinite structure. Analytic statements relating these two quantities are made in \cite{ammari2023spectral}. Here, we will briefly demonstrate that the generalized capacitance matrix can be used to relate the spectra of infinite and finite arrays of resonators. This is a non-trivial issue, since the two spectra have very different fundamental characteristics and there are complex edge effects occurring at the ends of the finite structure that need to be accounted for.

The basis of our comparison is a method that approximates the band structure of the periodic structure, given the set of eigenpairs $(\omega_j,u_j)$ of a truncated structure. If we take the size $r$ of the truncated structure to be reasonably large, then the eigenmode $u_j$ will \emph{approximately} be a linear combination of Bloch modes with frequency $\omega_j$. To compare the discrete eigenvalues of the truncated problem to the continuous spectrum of the periodic problem, we `reverse engineer' the appropriate quasi-periodicities $\alpha$ corresponding to these Bloch modes. Observe that an eigenvector $u_j$ is a vector of length $N|I_r|$. If we let $(u_j)_m$ denote the vector of length $N$ associated to cell $m\in \Lambda$, then we can define the \emph{truncated Floquet transform} of $u_j$ as
\begin{equation} \label{def:truncFloquet}
(\hat{u}_j)_\alpha = \sum_{m\in I_r} (u_j)_me^{\iu \alpha\cdot m}, \qquad \alpha \in  Y^*.
\end{equation}
Observe that $(\hat{u}_j)_\alpha$ is a vector of length $N$. Looking at the 2-norm $\|(\hat{u}_j)_\alpha\|_2$ as a function of $\alpha$, this function has distinct peaks at certain values of $\alpha$. We then take the quasi-periodicitiy associated to the mode $u_j$ as 
\begin{equation}
\argmax_{\alpha \in Y^*} \|(\hat{u}_j)_\alpha\|_2.
\end{equation}
Note that the symmetry of the problem means that if $\alpha$ is an approximate quasi-periodicity then so will $-\alpha$ be. In cases of additional symmetries of the lattice, we expect additional symmetries of the quasi-periodicities.

\begin{figure}
	\centering
	\includegraphics[width=0.55\linewidth]{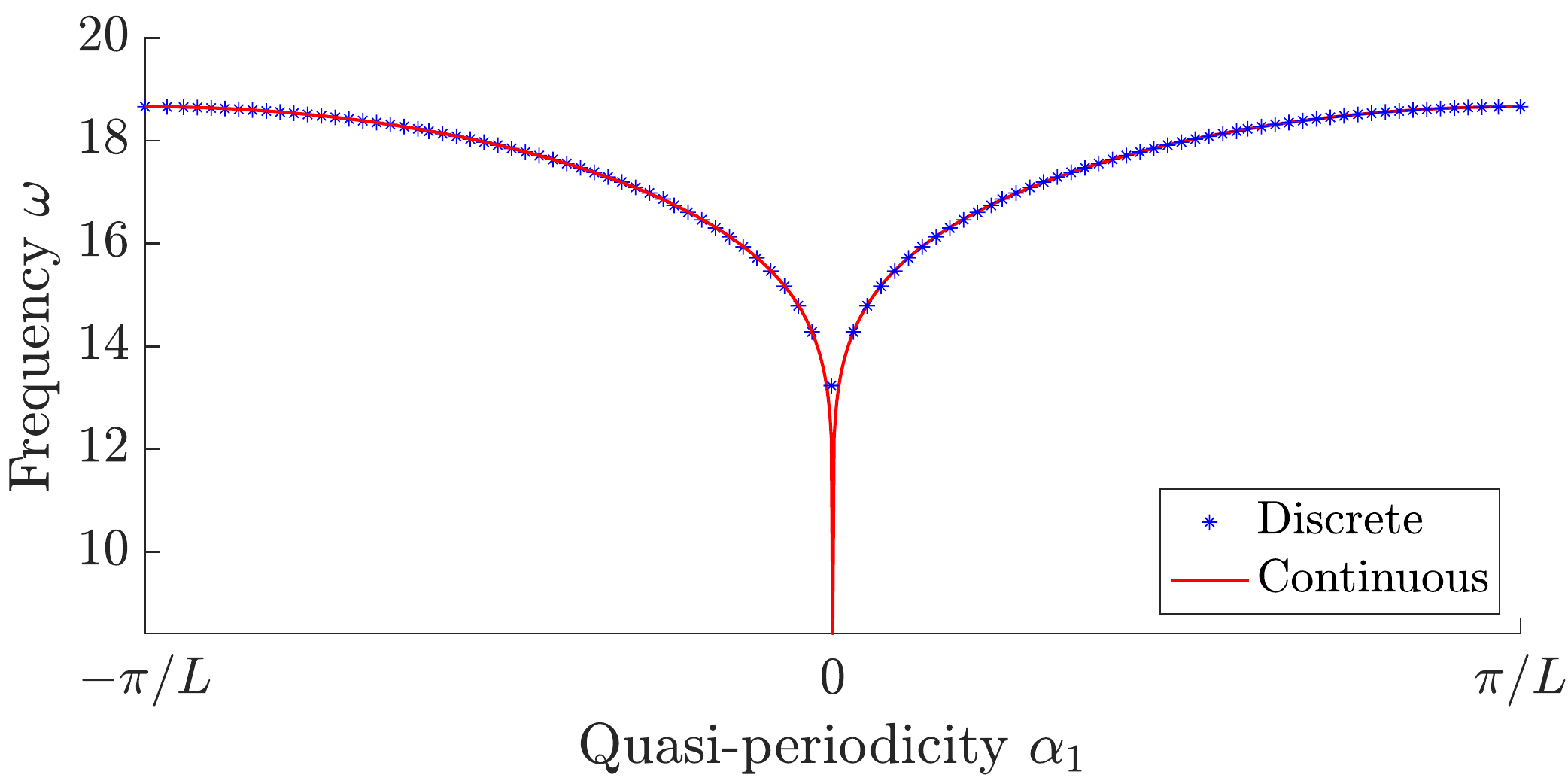}
	\begin{tikzpicture}
	\node at (0,0) {$\cdots$};
	\shade[ball color = gray!10] (0.5,0) circle (0.2);
	\draw (0.5,0) circle (0.2);
	\shade[ball color = gray!10] (1.2,0) circle (0.2);
	\draw (1.2,0) circle (0.2);
	\shade[ball color = gray!10] (1.9,0) circle (0.2);
	\draw (1.9,0) circle (0.2);
	\shade[ball color = gray!10] (2.6,0) circle (0.2);
	\draw (2.6,0) circle (0.2);
	\shade[ball color = gray!10] (3.3,0) circle (0.2);
	\draw (3.3,0) circle (0.2);
	\shade[ball color = gray!10] (4,0) circle (0.2);
	\draw (4,0) circle (0.2);
	\shade[ball color = gray!10] (4.7,0) circle (0.2);
	\draw (4.7,0) circle (0.2);
	
%	\draw[fill=white!80!gray] (0.5,0) circle(0.2);
%	\draw[fill=white!80!gray] (1.2,0) circle(0.2);
%	\draw[fill=white!80!gray] (1.9,0) circle(0.2);
%	\draw[fill=white!80!gray] (2.6,0) circle(0.2);
%	\draw[fill=white!80!gray] (3.3,0) circle(0.2);
%	\draw[fill=white!80!gray] (4,0) circle(0.2);
%	\draw[fill=white!80!gray] (4.7,0) circle(0.2);
	\node at (5.25,0) {$\cdots$};
	\node[white] at (0,-2.5) {.};
	\end{tikzpicture}
	\caption{The continuous spectrum of the infinite structure and the discrete spectrum of the truncated structure for a one-dimensional lattice with a single repeated resonator ($N=1$). The truncated structure has 50 resonators. The truncated Floquet transform \eqref{def:truncFloquet} is used to approximate the quasi-periodicity of the truncated modes.} \label{fig:band}
\end{figure}
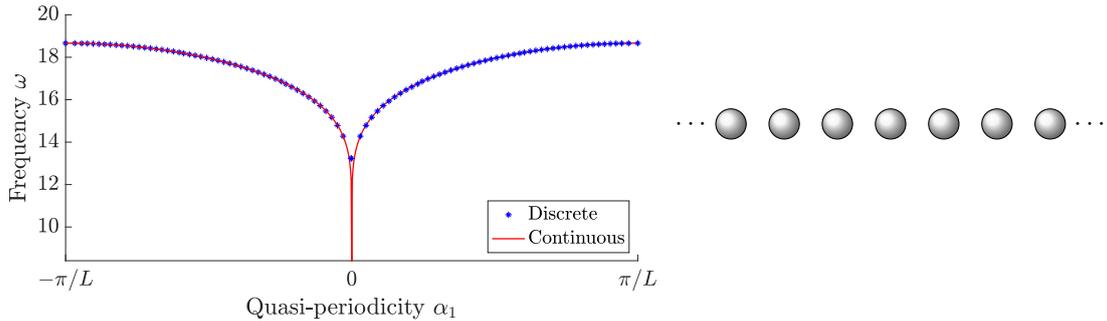

\Cref{fig:band} shows the subwavelength continuous spectrum of an infinite array of resonators, which takes the form of a single spectral band. It is plotted alongside the discrete spectrum of a truncated array of 50 resonators, for which the quasi-periodicities have been approximated using the method outlined above. The discrete band structure mostly follows closely the infinite one, even for this relatively small truncated array. The frequencies close to zero are not exhibited in the finite structure, as the edge effects have the greatest effect on low-frequency modes. We would need to consider a much larger truncated structure to capture the lowest frequency part of the spectrum. This behaviour can also be observed in more complicated structures, such as arrays of more than one repeating resonator ($N>1$) or multi-dimensional lattices, see \cite{ammari2023spectral} for details. 

\section{Concluding remarks}

In this work, we have demonstrated the convergence of defect modes in large resonator arrays to the corresponding modes in the infinite, periodic structure. We have studied this using the generalized capacitance matrix, which is a canonical model for three-dimensional wave scattering by resonant systems with long-range interactions. Our conclusions could also be generalized to other models, since the decay of the Helmholtz Green's function is the key feature that underpins our results. 

Our results clarify the exponential convergence of defect modes that was observed in previous studies \cite{lin2013resonances, lin2016perturbation, lu2022defect, lin2015twodim}. We observed that the exponential convergence occurs only when the dimension of the periodic lattice is equal to that of the differential problem. When the lattice has fewer dimensions than the space it is embedded in, the convergence is algebraic. This is due to the fact that waves are able to propagate away from the structure in the ``spare'' directions, leading to long-range interactions.

A significant advantage of the model used in this work is that the Bloch modes, in addition to the defect modes, are also concisely characterized. As detailed in \Cref{sec:cont}, this provides a numerical method for approximating the continuous spectrum. Importantly, this constructive approach presents a possible avenue for proving statements about the convergence of eigenvalues to the continuous spectrum, as we take advantage of in \cite{ammari2023spectral}.

%\section*{Data availability}
%
%No new data were created or analyzed in this study.

\appendix

\section{Asymptotic derivation of the generalized capacitance matrix} \label{sec:asymptotics}

In this brief appendix, we recall how the generalized capacitance matrix arises through an asymptotic treatment of a system of coupled high-contrast resonators. In particular, it can be used to characterize the subwavelength (\emph{i.e.} asymptotically low-frequency) resonance of the system. For more details and a review of extensions to other settings (such as non-Hermitian and time-modulated systems) see \cite{ammari2021functional}. 

We will present the results for a finite system of resonators. Analogous results hold for infinite periodic systems, by modifying the Green's function appropriately \cite{ammari2021functional}. It is also possible to understand defect modes in high-contrast systems using classical two-scale homogenization \cite{cherdantsev2009spectral, kamotski2018localized}. We suppose that the material inclusions $D_i\subset\mathbb{R}^3$, as considered already in this work, represent the material inclusions that will act as our resonators. We consider the scattering of time-harmonic waves with frequency $\omega$ and will solve a Helmholtz scattering problem in three dimensions. This Helmholtz problem, which can be used to model acoustic, elastic and polarized electromagnetic waves, represents the simplest model for wave propagation that still exhibits the rich phenomena associated to subwavelength physics.

We use $v_i$  denote the wave speed in each resonator $D_i$. In which case, $k_i=\omega/v_i$ is the wave number in $D_i$. Similarly, the wave speed and wave number in the background medium are denoted by $v$ and $k$. Finally, we must introduce the material contrast parameters $\delta_1,\dots,\delta_N$. These parameters describe the contrast between the material inside $D_i$ and the background material. For example, in the case of an acoustic system, $\delta_i$ is the density of the material inside $D_i$ divided by the density of the background material. We will want these contrast parameters to be small (an air bubble in water is one famous example in the setting of acoustics). Then for the domain 
$$D=\bigcup_{m\in I_r} \bigcup_{i=1}^N (D_i+m),$$
we consider the Helmholtz resonance problem
\begin{equation} \label{eq:finite_scattering}
	\left\{
	\begin{array} {ll}
		\ds \Delta {u}+ k^2 {u}  = 0 & \text{in } \R^3 \setminus \overline{D}, \\
		\nm
		\ds \Delta {u}+ k_i^2 {u}  = 0 & \text{in } D_i+m, \text{ for } i=1,\dots,N, \ m\in I_r, \\
		\nm
		\ds  {u}|_{+} -{u}|_{-}  =0  & \text{on } \partial D, \\
		\nm
		\ds  \delta_i \frac{\partial {u}}{\partial \nu} \bigg|_{+} - \frac{\partial {u}}{\partial \nu} \bigg|_{-} =0 & \text{on } \partial D_i+m \text{ for } i=1,\dots,N, \ m\in I_r, \\
		\nm
		\multicolumn{2}{l}{\ds u(x) \ \text{satisfies the Sommerfeld radiation condition},}
	\end{array}
	\right.
\end{equation}
where the Sommerfeld radiation condition says that
\begin{equation} \label{eq:SRC}
	\lim_{|x|\to\infty} |x|\left(\ddp{}{|x|}-\iu k\right)u=0, \quad \text{uniformly in all directions } x/|x|,
\end{equation}
and guarantees that energy is radiated outwards by the scattered solution. 

The asymptotic regime we consider is that the material contrast parameters are all small while the wave speeds are all of order one. That is, there exists some $\delta>0$ such that 
\begin{equation}
	\delta_i=O(\delta) \quad\text{and}\quad v,v_i=O(1) \quad \text{as} \quad \delta\to0, \text{ for } i=1,\dots,N.
\end{equation}
Within this setting, we are interested in solutions to the resonance problem \eqref{eq:finite_scattering} that are \emph{subwavelength} in the sense that
\begin{equation}
	\omega\to0 \quad \text{as}\quad \delta\to0.
\end{equation}

To be able to characterize the subwavelength resonant modes of this system, we must define the \emph{generalized} capacitance coefficients. Recall the capacitance coefficients $(C^{mn}_\frm)_{ij}$ from \eqref{eq:Cfinite}. Then, we define the corresponding generalized capacitance coefficient as
\begin{equation} \label{eq:gcm}
(\C_\frm^{mn})_{ij}=\frac{\delta_i v_i^2}{|D_i^m|} (C^{mn}_\frm)_{ij},
\end{equation}
where $|D_i^m|$ is the volume of the bounded subset $D_i^m$. Then, the eigenvalues of $\C_\frm$ determine the subwavelength resonant frequencies of the system, as prescribed by the following theorem.

\begin{thm}\label{thm:asymp}
	Consider a system of $N|I_r|$ subwavelength resonators in $\mathbb{R}^3$. For sufficiently small $\delta>0$, there exist $N|I_r|$ subwavelength resonant frequencies $\omega_1(\delta),\dots,\omega_{N|I_r|}(\delta)$ with non-negative real parts. Further, the subwavelength resonant frequencies are given by
	$$ \omega_n = \sqrt{\lambda_n}+O(\delta) \quad\text{as}\quad \delta\to0,$$
	where $\{\lambda_n: n=1,\dots,N|I_r|\}$ are the eigenvalues of the generalized capacitance matrix $\mathcal{C}_\frm$, which satisfy $\lambda_n=O(\delta)$ as $\delta\to0$.
\end{thm}

A similar result exists for an infinite periodic structure, in terms of the eigenvalues of the \emph{generalized} quasi-periodic capacitance matrix, as defined in \eqref{eq:Calpha}, see \cite{ammari2021functional} for details.

The definition \eqref{eq:gcm} clarifies the motivation for pre-multiplying by the perturbation matrix $\Bf$ to describe defects as in \eqref{eq:eig}. If we perturb the wave speed $v_i$ and contrast parameter $\delta_i$ such that 
$$\delta_iv_i^2 \to \delta_iv_i^2 b_i$$
for some coefficients $b_i$, the generalized capacitance coefficients are altered by multiplication by $b_i$. Similarly, the capacitance matrix $\Cf$ will be altered by pre-multiplication by $\Bf$. Observe that $b_i=1$ corresponds to  the unperturbed case, and the case of small perturbations is described by $b_i = 1+\eta_i$ for $|\eta_i| \ll 1$. When $\Bf$ is a compact perturbation of the identity, it describes defects that correspond to changing the material parameters on a finite number of resonators, so that the quantity $\delta_i v_i^2$ corresponding to those resonators is altered.

\section{Uniformity across the Brillouin zone}\label{sec:alp}
In this appendix, we provide additional details of the proof of \Cref{thm:approx}. The main result is \Cref{lem:bound}, which shows that $(\S_D^\alpha)^{-1}$ is in operator norm, uniformly bounded for $\alpha$ in a neighbourhood of $0$. The analysis is similar to \cite[Section 3.3]{ammari2020exceptional}.

From \emph{e.g.} \cite{ammari2009layer}, we have a dual-space representation of $G^\alpha$ given by
$$G^{\alpha}(x)= -\frac{1}{|Y|}\sum_{q\in \Lambda^*} \frac{e^{\iu(\alpha+q)\cdot x}}{ |\alpha+q|^2} =  \frac{-e^{\iu\alpha\cdot x}}{|Y| |\alpha|^2} -\frac{1}{|Y|}\sum_{q\in \Lambda^*\setminus\{0\}} \frac{e^{\iu(\alpha+q)\cdot x}}{ |\alpha+q|^2}.$$
Define the periodic Green's function $G^0$ as 
$$G^{0}(x)=-\frac{1}{|Y|}\sum_{q\in \Lambda^*\setminus\{0\}} \frac{e^{\iu q\cdot x}}{ |q|^2}.$$
For $\alpha$ close to zero, we then have 
$$G^{\alpha}(x)= \frac{-1}{ |Y||\alpha|^2} - \frac{\iu\alpha\cdot x}{ |Y||\alpha|^2} +\frac{(\alpha\cdot x)^2}{2|Y||\alpha|^2} + G^0(x) + O(|\alpha|).$$
Consequently, for $\alpha$ close to zero, we have the an expansion of the single-layer potential $\S_D^\alpha$:
\begin{multline}
	\S_D^\alpha[\psi](x) = -\frac{1}{|Y||\alpha|^2}\int_{\p D}\psi(y) \dx \sigma - \frac{\iu}{|Y||\alpha|^2}\int_{\p D}\alpha\cdot (x-y) \psi(y) \dx \sigma \\
	+ \frac{1}{2|Y||\alpha|^2}\int_{\p D}\bigl(\alpha\cdot (x-y)\bigr)^2 \psi(y) \dx \sigma + \S_D^0[\psi](x) + O(|\alpha|). \label{eq:exp}
\end{multline}
\begin{lem} \label{lem:phi=0}
	If $\S_D^{0}[\varphi] = K\chi_{\p D}$ for some constant $K$ and some $\varphi \in L^2(\p D)$ satisfying $\int_{\p D} \varphi \dx \sigma = 0$, then $\varphi = 0$.
\end{lem}
\begin{proof}
	For $x\in \R^3\setminus \D$, define $V(x) := \S_D^{0}[\varphi](x)$. Then $V$ solves the following differential problem,
	\begin{equation} \label{eq:V}
		\left\{
		\begin{array} {ll}
			\ds \Delta V  = 0 & \text{in } \R^3\setminus \D, \\[0.3em]
			\nm
			\ds  V|_{+}  = K  & \text{on } \partial \D, \\[0.3em]
			\nm
			\ds  V(x+m)  = V(x)  & \text{for all } m \in \Lambda.
		\end{array}
		\right.
	\end{equation}
	Moreover, using the jump relations and integration by parts, we have that 
	$$\int_{\p D} \varphi \dx \sigma = K\int_{Y \setminus D} |\nabla V|^2 \dx x=0.$$
	If $K\neq 0$, it follows from \eqref{eq:V} that $\int_{Y \setminus D} |\nabla V|^2 \dx x\neq0$ which is a contradiction. In other words we must have $K=0$, so that $\S_D^0[\varphi] = 0$ and $\int_{\p D} \varphi \dx \sigma = 0$. From \cite[Lemma~3.7]{ammari2020exceptional}, we have that $\varphi = 0$.
\end{proof}

\begin{lem} \label{lem:bound}
	$\|(\S_D^\alpha)^{-1}\|$, in operator norm, is bounded for $\alpha$ in a neighbourhood of $0$.
\end{lem}
\begin{proof}
	To reach a contradiction, we assume that $\S^\alpha_D[\phi] = O(|\alpha|)$ for some $\phi$, which can be written as $\phi = \phi_0 + |\alpha|\phi_1$, where $\phi_0$ is nonzero, does not depend on $\alpha$, and $\phi_1 = O(1)$ as $|\alpha|\to 0$. Also define $\v = \frac{\alpha}{|\alpha|}$. From \eqref{eq:exp} it follows that
	\begin{align*}
		\int_{\p D} \phi_0\dx \sigma &= 0, \\
		\int_{\p D}\phi_1(y) \dx \sigma +\iu\int_{\p D}\v\cdot (x-y) \phi_0(y) \dx \sigma &= O(|\alpha|), \\
		K(\v) - \iu\v\cdot x\int_{\p D} \phi_1(y) \dx \sigma + \frac{1}{2}\int_{\p D}\bigl(\v\cdot (x-y)\bigr)^2 \phi_0(y) \dx \sigma  + |Y|\S_D^0[\phi_0] &=O(|\alpha|),
	\end{align*}
	where $K$ is constant as function of $x$. Simplifying, we have that
	$$\frac{1}{2}\int_{\p D}\bigl(\v\cdot (x-y)\bigr)^2 \phi_0(y) \dx \sigma = -(\v\cdot x)\int_{\p D} (\v\cdot y) \phi_0(y)\dx\sigma  + \frac{1}{2}\int_{\p D} (\v\cdot y)^2\phi_0(y)\dx\sigma.$$
	In total we get
	$$\S_D^0[\phi_0](x) = \tilde K(\v) + \frac{2(\v\cdot x)}{|Y|}\int_{\p D} (\v\cdot y) \phi_0(y)\dx\sigma,$$
	where $\tilde K$ is constant in $x$. Observe that $\S_D^0[\phi_0](x)$ is independent of $\v$. As a function of $x$, this function is constant for $x\in \v^\perp$, and so this function is constant for all $x$. From \Cref{lem:phi=0} we get that $\phi_0 = 0$ which proves the claim.
\end{proof}

\section*{Acknowledgements}

The work of BD was supported by a fellowship funded by the Engineering and Physical Sciences Research Council under grant number EP/X027422/1. 

\bibliographystyle{abbrv}
\bibliography{fininf}{}
\end{document}